\documentclass[12pt]{article}
\usepackage[latin9]{inputenc}
\usepackage{amsmath}
\usepackage{amsthm}
\usepackage{amssymb}

\makeatletter

\providecommand{\tabularnewline}{\\}

\theoremstyle{plain}
\newtheorem{thm}{\protect\theoremname}
\theoremstyle{definition}
\newtheorem{defn}[thm]{\protect\definitionname}
\theoremstyle{plain}
\newtheorem{conjecture}[thm]{\protect\conjecturename}
\theoremstyle{definition}
\newtheorem{example}[thm]{\protect\examplename}
\theoremstyle{plain}
\newtheorem{lem}[thm]{\protect\lemmaname}
\theoremstyle{plain}
\newtheorem{prop}[thm]{\protect\propositionname}
\theoremstyle{theorem}
\newtheorem{claim}[thm]{\protect\claimname}
\theoremstyle{plain}
\newtheorem{cor}[thm]{\protect\corollaryname}
\theoremstyle{plain}
\newtheorem{fact}[thm]{\protect\factname}

\usepackage{fullpage}

\usepackage{amsthm}

\usepackage{bibspacing}

\makeatother

\providecommand{\claimname}{Claim}
\providecommand{\conjecturename}{Conjecture}
\providecommand{\corollaryname}{Corollary}
\providecommand{\definitionname}{Definition}
\providecommand{\examplename}{Example}
\providecommand{\factname}{Fact}
\providecommand{\lemmaname}{Lemma}
\providecommand{\propositionname}{Proposition}
\providecommand{\theoremname}{Theorem}

\begin{document}
\global\long\def\poly{\mathrm{{poly}}}%
\global\long\def\polylog{\mathrm{{polylog}}}%

\global\long\def\zo{\mathrm{\{0,1\}}}%

\global\long\def\mo{\mathrm{{-1,1}}}%

\global\long\def\e{\mathrm{\epsilon}}%

\global\long\def\d{\mathrm{\delta}}%

\global\long\def\a{\alpha}%

\global\long\def\b{\beta}%

\global\long\def\eps{\mathrm{\epsilon}}%

\global\long\def\z{\mathrm{\zeta}}%

\global\long\def\E{\mathrm{\mathbb{E}}}%

\global\long\def\F{\mathrm{\mathbb{F}}}%

\global\long\def\Z{\mathrm{\mathbb{Z}}}%

\global\long\def\R{\mathrm{\mathbb{R}}}%

\global\long\def\C{\mathrm{\mathbb{C}}}%

\global\long\def\P{\mathrm{\mathbb{P}}}%

\global\long\def\sign{\mathrm{{Sign}}}%

\global\long\def\maj{\mathrm{{Maj}}}%

\title{On correlation bounds against polynomials}
\author{Peter Ivanov{\normalsize{}{*}} \hspace{2cm} Liam Pavlovic$^{*}$
\hspace{2cm} Emanuele Viola\thanks{Supported by NSF grant CCF-2114116. Liam Pavlovic supported by Research
Experience for Undergraduates (REU) supplement.}}
\maketitle
\begin{abstract}
We study the fundamental challenge of exhibiting explicit functions
that have small correlation with low-degree polynomials over $\F_{2}$.
Our main contributions include:

1. In STOC 2020, CHHLZ introduced a new technique to prove correlation
bounds. Using their technique they established new correlation bounds
for low-degree polynomials. They conjectured that their technique
generalizes to higher degree polynomials as well. We give a counterexample
to their conjecture, in fact ruling out weaker parameters and showing
what they prove is essentially the best possible. 

2. We propose a new approach for proving correlation bounds with the
central ``mod functions,'' consisting of two steps: (I) the polynomials
that maximize correlation are symmetric and (II) symmetric polynomials
have small correlation. Contrary to related results in the literature,
we conjecture that (I) is true. We argue this approach is not affected
by existing ``barrier results.''

3. We prove our conjecture for quadratic polynomials. Specifically,
we determine the maximum possible correlation between quadratic polynomials
modulo 2 and the functions $(x_{1},\dots,x_{n})\to z^{\sum x_{i}}$
for any $z$ on the complex unit circle; and show that it is achieved
by symmetric polynomials. To obtain our results we develop a new proof
technique: we express correlation in terms of directional derivatives
and analyze it by slowly restricting the direction.

4. We make partial progress on the conjecture for cubic polynomials,
in particular proving tight correlation bounds for cubic polynomials
whose degree-3 part is symmetric.
\end{abstract}
%To do: Don't define d,a,b many times
%non-classical polys
%"relative distance" in the last proof is bad

\thispagestyle{empty}
\addtocounter{page}{-1}
\newpage

\section{Introduction and our results}

Exhibiting explicit functions that have small \emph{correlation} with
low-degree polynomials modulo 2 is a fundamental challenge in complexity
theory, cf. the recent survey \cite{corr-survey}. This challenge
is generally referred to as ``proving correlation bounds'' and progress
on it is a prerequisite for progress on a striking variety of other
long-standing problems: circuit lower bounds \cite{viola-FTTCS09,Viola-map},
Valiant's rigidity challenge \cite{Viola-diago}, number-on-forehead
communication complexity \cite{Viola-diago,Viola-map}, and even recently-made
conjectures on the Fourier spectrum of low-degree polynomials \cite{Viola-L12requiresCor}.

After many years, the state-of-the-art on this challenge has not changed
much since seminal works from at least thirty years ago. Two bounds
are known for degree $d$ polynomials. First, the results by Razborov
and Smolensky from the 80's give correlation $O(d/\sqrt{n})$ \cite{Raz87,Smo87,Smolensky93};
second, the result by Babai, Nisan, and Szegedy \cite{BNS92} on number-on-forehead
communication protocols yields correlation $\exp(-\Omega(n/d2^{d}))$.
A slight improvement to $\exp(-\Omega(n/2^{d}))$ appears in \cite{ViolaGF2}.
Thus, the first bound applies to large degrees but yields weak correlation,
while the second bound yields exponentially small correlation, but
only applies to degrees less than $\log n$. Achieving correlation
less than $1/\sqrt{n}$ for polynomials of degree $\log n$ remains
open, for any explicit function. Remarkably, solving this specific
setting of parameters is required for long-sought progress on any
of the challenges mentioned in the previous paragraph.

\subsection{The conjecture and our first result}

In STOC 2020, Chattopadhyay, Hatami, Hosseini, Lovett, and Zuckerman
\cite{DBLP:conf/stoc/ChattopadhyayHH20}. introduced a novel technique
which they established new correlation bounds for low-degree polynomials.
The key ingredient in their approach is a structural result about
the Fourier spectrum of low-degree polynomials over $\F_{2}$. They
show that for any $n$-variate polynomial $p$ over $\F_{2}$ of degree
$\le d$, there is a set $S$ of variables such that almost all of
the Fourier mass of $p$ lies on Fourier coefficients that intersect
with $S$, and the size of $S$ is exponential in $d$. Further, they
conjecture that the size of $S$ needs to be just polynomial in $d$.

We give a counterexample to their conjecture. In fact, we shall rule
out weaker parameters and show what they prove is essentially the
best possible. This appears in Section \ref{sec:CHHLZ-conjecture}.

\subsection{Mod functions}

A natural candidate for achieving small correlation are the $Mod_{\phi}$
functions which map inputs of Hamming weight $w$ to the complex point
on the unit circle with angle $w\phi$. These $Mod_{\phi}$ are closely
related to the boolean mod $m$ functions which indicate if the input
Hamming weight is divisible by $m$. Specifically, one can bound the
correlation with mod $m$ for odd $m$ by the correlations with the
$Mod_{\phi}$ functions for $\phi=2\pi k/m$ for $k=1,2,.\ldots,(m-1)/2$
(see Lemma \ref{prop:B-expr-general}). In turn, as discussed below,
an early motivation for studying the correlation with mod $m$ was
proving circuit lower bounds.

We now formally define these notions and then discuss previous results.
\begin{defn}
\label{def:c_phi}For any angle $\phi\in[0,2\pi]$ the function $Mod_{\phi}\colon\zo^{n}\to\mathbb{C}$
is defined as 
\[
Mod_{\phi}(x):=e^{\phi\sqrt{-1}\sum_{i}x_{i}}.
\]
The correlation of a polynomial $p:\zo^{n}\to\zo$ with $Mod_{\phi}$
is
\[
C_{\phi}(p):=\left|\E_{x\in\zo^{n}}(-1)^{p(x)}Mod_{\phi}(x)\right|.
\]
For any integer $m$ we define the boolean Mod $m$ function $BMod_{m}:\zo^{n}\to\mathbf{\zo}$
as

\[
BMod_{m}(x):=\begin{cases}
1 & \text{if }\sum_{i=1}^{n}x_{i}\neq0\bmod m\\
0 & \text{if }\sum_{i=1}^{n}x_{i}=0\bmod m.
\end{cases}
\]
The correlation between a polynomial $p:\zo^{n}\to\zo$ and $BMod_{m}$
is:

\[
B_{m}(p):=\left|\E_{x:BMod_{m}(x)=0}(-1)^{p(x)}-\E_{x:BMod_{m}(x)=1}(-1)^{p(x)}\right|.
\]
\end{defn}

Most or all of the works in this area, including this paper, is concerned
with the $Mod_{\phi}$ functions. And most of the works use correlation
bounds with $Mod_{\phi}$ functions for various $\phi$ to obtain
corresponding correlation bounds with the mod $m$ functions. In particular,
the two correlation bounds stated above hold for $Mod_{2\pi/3}$.
The first bound essentially appears in Smolensky's paper. For the
second bound, Bourgain first proved \cite{Bou05} correlation $\exp(-\Omega(n/c^{d}))$
with $Mod_{2\pi/m}$, with a correction in \cite{GRS05}. Nisan later
pointed out that such bounds also follow from \cite{BNS92}. The constant
$c$ is optimized to $4$ in \cite{ViolaGF2}. For more discussion
and background we refer to the survey \cite{corr-survey}, where the
reader may find proofs of both bounds, including Nisan's derivation
from \cite{BNS92}.

\subsubsection{Exact results }

Unlike other models of computation such as circuits, polynomials seem
simple enough that one may try to obtain \emph{exact} results\emph{.
}That is, one may try to precisely characterize the polynomials that
achieve the maximum correlation. Twenty years ago, a remarkable paper
by Green \cite{Gre04}, which is an inspiration for this work, took
precisely such a step. Green, and the subsequent work \cite{GreenRoy07},
precisely characterized the quadratic polynomials modulo \emph{three}
that achieve the maximum correlation with the $Mod_{2\pi/2}$ function,
i.e., parity. Compared to our discussion above, the moduli in \cite{Gre04}
are swapped. Green considers polynomials modulo 3 instead of 2, and
bounds the correlation with $Mod_{2\pi/2}$ instead of $Mod_{2\pi/3}$.
Extending Green's result to other moduli has resisted attacks, see
\cite{Gre04,DuenezMRS06}. While these works do not explicitly consider
polynomials modulo 2, difficulties also arise trying to port Green's
proof to our setting. In fact, jumping ahead, we will show that the
answer is different, arguably explaining the difficulties.

\subsubsection{Are symmetric polynomials optimal?}

Aiming for exact results, a natural question to ask is whether, for
some fixed degree, the polynomials modulo $2$ that have maximum correlation
with $Mod_{\phi}$ are \emph{symmetric. }Indeed, this question has
been asked by many authors; it appears explicitly for example in the
2001 paper by Alon and Beigel \cite{AlB01}. A positive answer would
have dramatic consequences since symmetric polynomials modulo 2, even
of large degree, have exponentially small correlation with, say, $Mod_{2\pi/3}$.
Thus, if one could prove that symmetric polynomials correlate best,
one would obtain long-sought correlation bounds.

However, until now the evidence for this has been negative. The maximizing
polynomials in \cite{Gre04,GreenRoy07} are \emph{not }symmetric.
Moreover, the work \cite{GreenKV-blocksym} has shown that for a large
range of parameters, symmetric polynomials modulo 3 do \emph{not}
correlate best with parity (and are not even close). One of the families
of polynomials that are shown to outperform symmetric in these works
is that of \emph{block-symmetric} polynomials, which are sums of symmetric
polynomials on disjoint sets of variables. However, naive conjectures
regarding the optimality of block-symmetric or other families of polynomials
fail, and we are not aware of any natural family of polynomials modulo
3 that is a candidate to maximizing correlation with parity. The only
available evidence that symmetric polynomials correlate best with
mod functions are computer experiments up to 10 variables reported
in \cite{GreenKV-blocksym}.

\subsection{A new approach}

Departing from previous proofs, in this work we propose the following
approach to proving correlation bounds with mod functions. It consists
of two steps:

(I) Prove that symmetric polynomials correlate best with mod functions,
and

(II) Prove that symmetric polynomials have exponentially small correlation
with mod functions.

Regarding (I), we put forth the following conjecture:
\begin{conjecture}
\label{conj:sym-best}For every $d,n,\phi$ degree-$d$ symmetric
polynomials correlate best with the $Mod_{\phi}$ function on $n$
bits.
\end{conjecture}

We verify (II) in Section \ref{sec:symm_vs_mod}. The result is folklore.
We remark that \cite{CaiGT96} proves a similar result, but in the
case of symmetric polynomials mod $m$ and the mod 2 function. However,
changing moduli can yield different results, as shown by this paper.

\subsubsection{Our approach vs.~``barriers'' to lower bounds}

Over the years many ``barriers'' have been proposed for progress
on lower bounds. Barriers based on oracles or relativization \cite{BGS75,AaronsonW08}
are not known to apply -- they mostly concern uniform models of computation.
The Natural Proofs barrier \cite{RaR97} (see also \cite{NaorRR02,MilesV-aes})
is also not known to apply since we do not have candidate pseudorandom
functions that correlate with low-degree polynomials.

More recently, Bhowmick and Lovett \cite{BhowmickL15} proposed a
new barrier specifically for proving correlation bounds. They consider
an extension of polynomials called \emph{non-classical polynomials,
}an object first introduced in \cite{TZ12}\emph{.} In short, in a
non-classical polynomial of degree $d$ monomials can have rational
coefficients (with denominators depending on the degree) and the output
of the polynomial is considered as an element in the torus $[0,1]$.
The work \cite{BhowmickL15} shows that the proofs of most correlation
bounds (such as those mentioned at the beginning of this introduction)
also apply to non-classical polynomials. Moreover, for non-classical
polynomials these bounds are actually tight! For example, there are
non-classical polynomials of degree just $O(\log n)$ that correlate
well with mod functions.

We argue that non-classical polynomials do not constitute an obstacle
for our approach above. The main reason is that the non-classical
polynomials in \cite{BhowmickL15} -- including those for mod functions
-- are actually symmetric. Hence, one could conceivably prove (I)
above without distinguishing classical from non-classical polynomials.
Moreover, the proof of (II) above already distinguishes classical
from non-classical polynomials.

\subsection{Our second result: Proof of Conjecture \ref{conj:sym-best} for $d=2$}

A main technical contribution of this work is a proof of our Conjecture
\ref{conj:sym-best} in the case of degree two. That is, in contrast
with the previous proofs discussed above, we show that, among quadratic
polynomials modulo $2$, those that correlate best with the $Mod_{\phi}$
functions are symmetric. Let us first define the elementary symmetric
polynomials of degree 1 and $2$.
\begin{defn}[Elementary symmetric polynomials]
 Let
\begin{align*}
e^{1}(x_{1},\ldots,x_{n}) & :=\sum_{i=1}^{n}x_{i},\\
e^{2}(x_{1},\ldots,x_{n}) & :=\sum_{i<j}^{n}x_{i}x_{j}.
\end{align*}
\end{defn}

\begin{example}
\label{exa:cor} Let $\phi=2\pi/3$ and $\omega=e^{\phi\sqrt{-1}}$.
We have:
\begin{align*}
C_{\phi}(0) & =\left|\mathop{\E}\limits _{x\in\zo^{n}}\omega^{\sum_{i}x_{i}}\right|=\left|\mathop{\E}\limits _{x_{1}\in\zo}\omega^{x_{1}}\right|^{n}=\left|\frac{1+\omega}{2}\right|^{n}=\left(\frac{1+\cos\phi}{2}\right)^{n/2}=\left(\frac{1}{2}\right)^{n},\\
C_{\phi}(e^{1}) & =\left|\mathop{\E}\limits _{x\in\zo^{n}}(-1)^{\sum_{i}x_{i}}\omega^{\sum_{i}x_{i}}\right|=\left|\mathop{\E}\limits _{x_{1}\in\zo}(-1)^{x_{1}}\omega^{x_{1}}\right|^{n}=\left|\frac{1-\omega}{2}\right|^{n}=\left(\frac{1-\cos\phi}{2}\right)^{n/2} \\ & =\left(\frac{\sqrt{3}}{2}\right)^{n},\\
C_{\phi}(BMod_{3}) & \ge1/2,
\end{align*}
where the last inequality follows because the absolute value of the
real component of \\ $(-1)^{BMod_{3}(x)}\omega^{\sum_{i}x_{i}}$ is $\ge1/2$
for every $x$.
\end{example}

We next state our result. Henceforth all polynomials in this paper
have coefficients in $\zo$ and operate modulo two. We characterize
the quadratic polynomials that maximize $C_{\phi}$ \emph{for any
angle} $\phi\in[0,2\pi]$. Additionally, we show the correlation of
other quadratic polynomials is a multiplicative factor smaller.

It is in fact sufficient to restrict our attention to angles $\phi\in[0,\pi/2]$
thanks to a simple symmetry argument presented in Section \ref{sec:Derivatives}.
When $\phi\in[0,\pi/4]$ then the constant zero polynomial maximizes
correlation. Our main contribution is that when $\phi\in(\pi/4,\pi/2]$
the correlation is maximized by either $e^{2}$ or $e^{2}+e^{1}$,
depending on the value of $n\bmod4$.

We define the quantity
\[
v_{\phi}:=2^{-n-1}\cdot\left((1+\sin\phi)^{n}+(1-\sin\phi)^{n}\right)
\]
which plays a key role in this paper.
\begin{thm}
\label{thm:main-general} Fix any angle $\phi\in[0,\pi/2].$ For all
large enough $n$, the maximum $C_{\phi}(p)$ over quadratic polynomials
$p$ is attained by a symmetric polynomial. In more detail:
\begin{enumerate}
\item Suppose $\phi\in(\pi/4,\pi/2]$.
\begin{enumerate}
\item For $n$ even we have $C_{\phi}(e^{2})=C_{\phi}(e^{2}+e^{1})=\sqrt{v_{\phi}}$.
\item For $n\equiv1\bmod4$ we have $C_{\phi}(e^{2})=\sqrt{v_{\phi}+(\cos(\phi)/2)^{n}}$,
$C_{\phi}(e^{2}+e^{1})=\sqrt{v_{\phi}-(\cos(\phi)/2)^{n}}$.
\item For $n\equiv3\bmod4$ we have $C_{\phi}(e^{2})=\sqrt{v_{\phi}-(\cos(\phi)/2)^{n}}$,
$C_{\phi}(e^{2}+e^{1})=\sqrt{v_{\phi}+(\cos(\phi)/2)^{n}}$.
\item For any quadratic polynomial $p$ besides $e^{2}$, $e^{2}+e^{1}$
we have \\
 $C_{\phi}(p)\le\sqrt{1-\Omega(\sin\phi-\cos\phi)}\cdot\sqrt{v_{\phi}}.$
\end{enumerate}
\item Suppose $\phi\in[0,\pi/4]$. Then $C_{\phi}(0)=\left(\frac{1+\cos\phi}{2}\right)^{n/2}$
and for any quadratic polynomial $p\ne0$ we have $C_{\phi}(p)\le(1-\Omega(1))\cdot C_{\phi}(0).$
\end{enumerate}
\end{thm}

Note that $\sqrt{v_{\phi}-(\cos(\phi)/2)^{n}}\ge(1-o(1))\sqrt{v_{\phi}}$
and so the theorem shows that the correlation of non-symmetric polynomials
is a constant-factor smaller than optimal.

An important message of this paper is that $C_{\phi}$ is maximized
by \emph{symmetric} polynomials. This contrasts with previous works,
and gives hope that this may hold for larger degrees as well. If that
is the case one would obtain long-sought correlation bounds, as discussed
previously.

\subsubsection{Results and directions for $d=3$}

We conjecture that Theorem \ref{thm:main-general} can be extended
to show that for any cubic polynomial $p$ and any $\phi$, $C_{\phi}(p)\leq\max_{s\in\{0,e^{1},e^{2},e^{2}+e^{1}\}}C_{\phi}(s)$.
In other words, the correlation over all cubic polynomials is still
maximized by a quadratic symmetric. This would prove Conjecture \ref{conj:sym-best}
for $d=3$ as well. 

We make progress on this conjecture by proving this indeed holds when
$p$ is the sum of an arbitrary quadratic polynomial and a symmetric
degree-3 polynomial. This is done in Section \ref{sec:structured-cubic}.

\subsection{Boolean correlation}

We now turn our attention to the boolean $BMod_{m}$ function. As
mentioned earlier, most or all papers bounding the corresponding correlation
$B_{m}$, including this one, proceed by first bounding $C_{\phi}$
for several corresponding values of $\phi$ and then using that information
to bound $B_{m}$. Indeed, $C_{\phi}$ is a better-behaved quantity
to work with. In turn, an early motivation for studying $B_{m}$ is
the so-called \emph{discriminator lemma} \cite{HMPST93}. The lemma
implies that if there is a circuit consisting of a majority of $s$
functions that computes $BMod_{m}$ then one of those functions $p$
has $B_{m}(p)\ge1/s$. Thus, one can use upper bounds on $B_{m}$
to obtain lower bounds for such circuits.

In this paper we determine up to constant factors the maximum of $B_{m}$
over quadratic polynomials. This is Item 1 in the next theorem. In
fact, we obtain more precise information. Item 2 determines (exactly)
the maximum value when $n$ is congruent to $m,3m\bmod4m$: either
$e^{2}$ or $e^{2}+e^{1}$ maximizes $B_{m}$, and moreover it will
achieve the upper bound on $B_{m}$ from Item 1. Our inability to
determine the maximum value of $B_{m}$ for every $n$ is reflected
in Item 3, which shows when $n$ is congruent to $0,2m\bmod4m$ this
maximum is not achieved by symmetric polynomials.
\begin{thm}
\label{thm:boolean-theorem-modm}Fix any odd $m\geq3$, let $\phi:=2\pi/m$,
$\ell_{1}\in\{\frac{m-1}{4},\frac{m+1}{4}\}$ denote the integer closest
to $\frac{m}{4}$, and set $\Psi:=2m/(m-1)\sqrt{v_{\ell_{1}\phi}}$.
The following holds for large enough $n$. Let $B_{m}^{*}$ denote
the maximum $B_{m}(p)$ over all quadratic $p$.
\begin{enumerate}
\item For any $n$,
\[
\Psi(1/\sqrt{2}-o(1))\leq B_{m}^{*}\leq\Psi(1+o(1)).
\]
\item If $n\equiv m,3m\bmod4m$ then
\[
B_{m}^{*}=\max_{s\in\{e^{2},e^{2}+e^{1}\}}B_{m}(s)=\Psi(1-o(1)).
\]
\item If $n\equiv0,2m\bmod4m$ then
\[
(1+\Omega(1))\max_{s\in\{0,e^{1},e^{2},e^{2}+e^{1}\}}B_{m}(s)<\max_{s'\in\{e^{2},e^{2}+e^{1}\}}B_{m}(x_{1}+s'(x_{2},\dots,x_{n})).
\]
\end{enumerate}
\end{thm}

Note that the polynomial in the right-hand side of Item 3 is not symmetric.
We conjecture that this polynomial is in fact optimal (for the corresponding
values of $n$). Our techniques yield slightly stronger results for
specific $m$ and $n$, but for simplicity we only state the above
theorem that applies for any odd $m\geq3$. In particular, when $m=3$,
it is possible to determine for every value of $n$ whether symmetric
polynomials maximize $B_{3}$.

Previous techniques could at best determine this maximum up to polynomial
factors. Hence we also improve polynomially the corresponding circuit
lower bounds obtained via the discriminator lemma -- this is a straightforward
application that we do not state formally.

Green's work \cite{Gre04} also determines exactly the maximum correlation
between quadratic polynomials modulo 3 and the parity function. Our
setting appears somewhat complicated by the fact that the $BMod_{m}$
functions are not balanced for odd $m$.

\subsection{Proof sketch of Theorem \ref{thm:main-general}}

We begin by rewriting the correlation in a more convenient form, involving
derivatives of the polynomial and of the mod function. Bounding the
correlation in terms of derivatives is natural and done in several
previous works, see e.g. discussion of the `squaring trick' in \cite[Chapter 1]{viola-SIGACT09}.
However, these works take repeated derivatives until the polynomial
becomes constant, use the Cauchy-Schwartz inequality, and are lossy. 

By contrast, we take a single derivative, avoid Cauchy-Schwartz, and
give an exact expression. In other words, previous works provide asymptotic
correlation bounds for larger degree polynomials, while we provide
an \emph{exact bound }for quadratic polynomials. 

For concreteness consider the complex mod 3 function $Mod_{\phi}:=e^{\phi\sqrt{-1}\sum_{i}x_{i}}:=\omega^{\sum_{i}x_{i}}$
where $\phi:=2\pi/3$, and fix some quadratic $p$. Let $p_{y}$ denote
the derivative $p(x+y)+p(x)$ of $p$ in the direction $y\in\zo^{n}$.
Analogously we let $Mod_{\phi,y}(x):=\omega^{\sum_{i}x_{i}-\sum_{i}(x_{i}\oplus y_{i})}$.
We can express the correlation squared as
\begin{align*}
C_{\phi}^{2}(p) & =\E_{y}\E_{x}(-1)^{p_{y}(x)}Mod_{\phi,y}(x).
\end{align*}
Writing $c_{y}(p)$ for the inner expectation -- where $c$ stands
for \emph{contribution} in direction $y$ -- we express the above
as $\E_{y}c_{y}(p)$. In this language, our goal now is to prove the
following for any quadratic $p$ and $s=e^{2},e^{2}+e^{1}$:
\begin{equation}
\E_{y}|c_{y}(p)|\leq\E_{y}c_{y}(s).\label{eq:disc-open-first-1}
\end{equation}

\subsubsection{Computing $\protect\E_{y}c_{y}(s)$ and bounding $|c_{y}(p)|$}

We begin by deriving a clean expression for $\E_{y}c_{y}(s)$. Let
$w(y)$ denote the Hamming weight of $y$ and let $E,O$ denote the
set of even, odd weight strings respectively. Supposing $n$ is even
for simplicity we have:
\begin{align}
\E_{y}c_{y}(s) & =2^{-n}\sum_{y\in E}(\sin\phi){}^{w(y)}.\label{eq:disc-symm-corr}
\end{align}
To see this, observe that $s_{y}=\sum_{i:y_{i}=1}x_{i}$ if $y\in E$
and $s_{y}=\sum_{i:y_{i}=0}x_{i}$ if $y\in O$. On the other hand,
$Mod_{\phi,y}=\omega^{\sum_{i:y_{i}=1}(2x_{i}-1)}$ which only depends
on the variables indexed by the 1 bits of $y$ \emph{for every }$y$. 

This means that for any $y\in O$, $c_{y}(s)=0$ and for any $y\in E$, $c_{y}(s)=(\sin\phi)^{w(y)}$. Together this implies (\ref{eq:disc-symm-corr}).

Moreover, by observing that $p_y(x)$ is linear one can show that $(\sin\phi)^{w(y)}$  is in fact an upper bound on $|c_{y}(p)|$. In other words, for any quadratic $p$
and direction $y$ we have
\begin{equation}
|c_{y}(p)|\leq(\sin\phi)^{w(y)}.\label{eq:disc-trivial-max-1}
\end{equation}
This is an important fact we will use throughout the proof.

\subsubsection{Structure on $p$ and slowly restricting $y$}
To deal with $\sum_{y}|c_{y}(p)|$, we will first illustrate how we can bound $\sum_{y:y_1=0}|c_{y}(p)|$. Looking ahead, we are able to deal with any partial sum where at least one bit in $y$ is restricted to
0, as long as $p$ possesses certain structure. This idea, combined with one more ingredient we discuss in the next section, is the heart of the main proof.  

For the sake of simplicity, suppose that $p=x_{1}x_{2}+q(x_{3},\dots,x_{n})$
for some quadratic $q$. With this structure on $p$, it turns out
we gain something after conditioning on $y_{1}=0$:
\begin{equation}
\sum_{y:y_{1}=0}|c_{y}(p)|\leq\sum_{y:y_{1}=0,y_{2}=0}(\sin\phi)^{w(y)}.\label{eq:first-easy-bound}
\end{equation}
We gain since this improves on the the bound which follows by only using  (\ref{eq:disc-trivial-max-1}):
\begin{equation}
\sum_{y:y_{1}=0}|c_{y}(p)|\leq\sum_{y:y_{1}=0}(\sin\phi)^{w(y)}.
\end{equation}

To prove (\ref{eq:first-easy-bound}) we condition on $y_{2}$. If
$y_{2}=1$ then we show $c_{y}(p)=0$ by mimicking the proof that
$c_{y}(s)=0$ for any $y\in O$. By assumption on $p$ we have $p_{y}(x)=x_{1}+q_{y'}(x')$
for any $y=01y'$. And recall $Mod_{\phi,y}=\omega^{\sum_{i:y_{i}=1}(2x_{i}-1)}$
does not depend on $x_{1}$ since $y_{1}=0$. If $y_{2}=0$ then we
use the bound from (\ref{eq:disc-trivial-max-1}). Combining the two
cases implies (\ref{eq:first-easy-bound}). 

In the next step, we would ideally like to bound $\sum_{y:y_{1}=1}|c_{y}(p)|$.
However, it is not clear how to repeat the previous step, where
the assumption on $p$ and restricting $y_{1}=0$ crucially allowed
us to observe that $c_{y}(p)=0$ for half the directions.

To overcome this, we instead condition on $y_{1}=1,y_{2}=0$. Now
$p_{y}(x)=x_{2}+q_{y'}(x')$ for any $y=10y'$, but $Mod_{\phi,y}$
does not depend on $x_{2}$ since $y_{2}=0$. Hence
\[
\sum_{y:y_{1}=1,y_{2}=0}|c_{y}(p)|=0.
\]

To summarize, we can make progress on the partial sum $\sum_{y:y_{1}=1,\dots,y_{j-1}=1}|c_{y}(p)|$
by conditioning on $y_{j}=0$, as long $x_{j}$ has certain structure
in $p$. This argument gives a non-trivial bound on $\sum_y |c_y(p)|$,
but is still not enough to prove (\ref{eq:disc-open-first-1}). We
strengthen it in the next section.

\subsubsection{Bounding $\protect\E_{y}|c_{y}(p)|$}

We are almost ready to prove our initial goal:
\[
\E_{y}|c_{y}(p)|\leq\E_{y}c_{y}(s).
\]
The last ingredient we need is that (\ref{eq:disc-trivial-max-1})
can be improved to
\begin{equation}
|c_{y}(p)|\leq(\sin\phi)^{w(y)-1}(\cos\phi)\label{eq:disc-handshaking-max-1}
\end{equation}
whenever $y\in O$, which we sketch in the next section.

Our proof strategy is similar to that of the previous section. We restrict
the direction one bit at a time, but now,
we will directly compare $\sum|c_{y}(p)|$ to $\sum c_{y}(s)$. In
the first step we show that
\begin{equation}
\sum_{y:y_{1}=0}|c_{y}(p)|\leq\sum_{y:y_{1}=0}c_{y}(s).\label{eq:second-harder-bound}
\end{equation}
We bound $\sum_{y:y_{1}=0}|c_{y}(p)|$ by applying (\ref{eq:disc-handshaking-max-1})
for the odd weight directions, which allows us to improve the bound
on $\sum_{y:y_{1}=0}|c_{y}(p)|$ from (\ref{eq:first-easy-bound})
to the following:
\[
\sum_{y:y_{1}=0}|c_{y}(p)|\leq\sum_{y:y_{1}=0,y_{2}=0,y'\in E}(\sin\phi)^{w(y)}+\sum_{y:y_{1}=0,y_{2}=0,y'\in O}(\sin\phi)^{w(y)-1}\cos\phi.
\]
To compare this to $\sum_{y:y_{1}=0}c_{y}(s)$, we recall the expression
from (\ref{eq:disc-symm-corr}) which implies
\[
\sum_{y:y_{1}=0}c_{y}(s)=\sum_{y:y_{1}=0,y_{2}=0,y'\in E}(\sin\phi){}^{w(y)}+\sum_{y:y_{1}=0,y_{2}=1,y'\in O}(\sin\phi){}^{w(y)}.
\]
Now we can conclude the proof of (\ref{eq:second-harder-bound}) as
\[
\sum_{y:y_{1}=0,y_{2}=0,y'\in O}(\sin\phi)^{w(y)-1}\cos\phi\leq\sum_{y:y_{1}=0,y_{2}=1,y'\in O}(\sin\phi){}^{w(y)}.
\]
We remark the improvement from (\ref{eq:disc-handshaking-max-1})
is crucial since if we just used (\ref{eq:first-easy-bound}) then
we would need 
\[
\sum_{y:y_{1}=0,y_{2}=0,y'\in O}(\sin\phi)^{w(y)}\leq\sum_{y:y_{1}=0,y_{2}=1,y'\in O}(\sin\phi){}^{w(y)}
\]
which is clearly false as $\sin\phi<1$.

For the next step, assuming that $x_{2}$ appears in at least a few
quadratic terms (for the precise conditions see Lemmas \ref{lem:odd-even-general},
\ref{lem:gap-general}), we can similarly show that 
\[
\sum_{y:y_{1}=1,y_{2}=0}|c_{y}(p)|\leq\sum_{_{y:y_{1}=1,y_{2}=0}}c_{y}(s).
\]

We continue this process until there are no more suitable direction
bits to condition on. When this happens, we conclude by reasoning
on the remaining structure of the polynomial (see Lemmas \ref{lemma:done-if-most-open},
\ref{lemma:low-degree-loses}).

\subsubsection{The proof of (\ref{eq:disc-handshaking-max-1}) via handshaking}

For any $p$ and $y$, we can determine $p_{y}(x)$ by examining the
graph $G_{p,y}$, which is defined with $w(y)$ nodes that correspond
to the variables indexed by the 1 bits of $y$, and edges that represent
the quadratic terms of $p$ on those $w(y)$ variables. Observe that
$x_{i}$ appears in $p_{y}(x)$ iff $x_{i}$ has odd degree in $G_{p,y}$. 

Now fix some $y\in O$. The number of nodes in $G_{p,y}$ is odd,
and the handshaking lemma implies the number of nodes in $G_{p,y}$
with odd degree must be even. Together this implies $p_{y}(x)$ contains
at most $w(y)-1$ variables which in turn implies (\ref{eq:disc-handshaking-max-1})
after a calculation. For the formal proof see Claim \ref{claim:contribution-is-even}.

\subsubsection{Slackness}

Although we get exact results in the end, we emphasize that some steps
in the proof do not yield exact bounds, but are approximate. For example,
after we open the first bit we in fact show a strict inequality between
$\E_{y:y_{1}=0}|c_{y}(p)|$ and $\E_{y:y_{1}=0}c_{y}(s)$ when $p$
is non-symmetric (Lemma \ref{lem:buffer-general}). This gives us
a ``buffer'' between $\E_{y}|c_{y}(p)|$ and $\E_{y}c_{y}(s)$,
which is reflected in the statement of Item 1(d) in Theorem \ref{thm:main-general}.

This extra factor is not just additional information, but is in fact
critical for the proof since the final step might be lossy (this occurs
when Lemma \ref{lemma:done-if-most-open} is applied). The buffer
gained will be much larger than the loss from Lemma \ref{lemma:done-if-most-open}
which allows us to conclude the proof.

\section{The CHHLZ conjecture \label{sec:CHHLZ-conjecture}}

In this section we present the new technique in \cite{DBLP:conf/stoc/ChattopadhyayHH20},
their conjecture, and our counterexample. The key ingredient in the
approach in \cite{DBLP:conf/stoc/ChattopadhyayHH20} is a structural
result about the Fourier spectrum of low-degree polynomials over $\F_{2}$.
They show that for any $n$-variate polynomial $p$ over $\F_{2}$
of degree $\le d$, there is a set $S$ of variables such that almost
all of the Fourier mass of $p$ lies on Fourier coefficients that
intersect with $S$, and the size of $S$ is exponential in $d$.
This remarkable result allows them to prove new correlation bounds.
Further, they conjecture that the size of $S$ needs to be just polynomial
in $d$.

Next we present their conjecture in more detail, and then our results.
The main quantity used in \cite{DBLP:conf/stoc/ChattopadhyayHH20}
is ``local correlation'' which they define as follows:
\begin{defn}[Local correlation, \cite{DBLP:conf/stoc/ChattopadhyayHH20}]
 For any $F:\zo^{n}\to\{-1,1\}$,
\[
\Delta_{S}(F):=\E_{x^{\overline{S}}}\left[(\E_{x^{S}}[F(x)]-\E[F])^{2}\right].
\]
\end{defn}

For a polynomial $p:\F_{2}^{n}\to\F_{2}$ we write $e(p)$ for $(-1)^{p}$
which takes values in $\{-1,1\}$. Next we state their conjecture:
\begin{conjecture}[{\cite[Conjecture 1.14]{DBLP:conf/stoc/ChattopadhyayHH20}}]
\label{conj:chhlz-1}For every polynomial $p$ of degree $d$ there
exists a set $S$ of $\leq poly(d,\log(1/\e))$ variables such that
$\Delta_{S}(e(p))\leq\eps$.
\end{conjecture}

In fact CHHLZ make a stronger conjecture (Conjecture 1.15 in \cite{DBLP:conf/stoc/ChattopadhyayHH20}),
where a single set $S$ is found that works for an entire space of
dimension $k$ of polynomials. This generality is critical in proving
their new correlation bounds. However, we shall give a counterexample
even for $k=1$. In fact, we shall rule out even much weaker parameters
and show that what they prove is essentially the best possible. Specifically,
we show that for $d=O(\log n)$ and constant $\e$, one needs $|S|\ge n/\log^{O(1)}n$.
\begin{thm}
\label{thm:counterex-CHLLZ}There exists a polynomial $p$ of degree
$d=O(\log n)$ such that $\Delta_{S}(e(p))\ge\Omega(1)$ for any $S$
of size $\leq c\cdot n/\log^{2}n$, where $c>0$ is an absolute constant.
\end{thm}

The rest of this section is devoted to the proof of this theorem.
The idea behind it is quite natural in hindsight, and highlights the
expressive power of polynomials of degree $O(\log n)$. 
\begin{defn}
\cite{Ben-OrLi85} (cf. \cite{ODonnell-boolean2007}, Proposition
4.12) We define $\mathrm{TRIBES}:\zo^{n}\to\zo$ to be a read-once
monotone DNF where every term has size $w$ so that $|\E_{x}[\mathrm{TRIBES}(x)]-1/2|\leq O(\log n)/n$.
This makes $w=\log n-\log\log n+O(1)$.
\end{defn}

The next result shows the probability $\mathrm{TRIBES}$ is fixed
to 1 after a uniform assignment to $x^{\overline{S}}$ is approximately
the same as after a uniform assignment to $x$, where $S\subset[n]$
is a subset of nearly linear size. This property was also used in
\cite{IP-AC0Parity} to show separations between DNFs composed with
parity gates and parity decision trees.
\begin{lem}
\label{lem:tribes-resilient}Fix any $S\subset[n]$ such that $|S|\leq O(n/\log^{2}n)$.
Then
\[
\P_{x^{\overline{S}}}[\mathrm{\mathrm{TRIBES}}(x)\text{ not fixed }]\leq1/2+o(1).
\]
\end{lem}

\begin{proof}
The set $S$ can intersect at most $|S|$ $\mathrm{AND}$ terms. The
probability over a uniform assignment to $x^{\overline{S}}$ that
$\mathrm{TRIBES}(x)$ is fixed to 1 is at least the probability one
of the untouched $\mathrm{AND}$ terms is set to 1. Hence,
\begin{align*}
\P_{x^{\overline{S}}}[\mathrm{TRIBES}(x)=1] & \geq1-(1-2^{-w})^{n/w-|S|}.\\
 & =1-\frac{\P_{x}[\mathrm{TRIBES}(x)=0]}{(1-2^{-w})^{|S|}}\\
 & \geq1-(1/2+O(\log n)/n)(1+1/\Omega(\log n))\\
 & \geq1/2-1/\Omega(\log n).
\end{align*}
where the second $\geq$ follows since $(1-2^{-w})^{|S|}\geq1-|S|/2^{w}\geq1-1/\Omega(\log n)$
and the fact $1/(1-x)\geq1+x$.
\end{proof}
We next show that $\mathrm{TRIBES}$ can be approximated by a low-degree
polynomial. This can be seen as a special case of Razborov's classical
approximation \cite{Raz87}.
\begin{lem}
\label{cor:lem-degree-resilient}There exists a $O(\log n)$ degree
polynomial $p$ such that 
\[
\E_{x}[e(\mathrm{TRIBES}(x))e(p(x))]\geq1/2+\Omega(1).
\]
\end{lem}

\begin{proof}
We will construct a distribution $D$ of $O(\log n)$ degree polynomials
such that for any $x$, $\P_{q\sim D}[q(x)\neq\mathrm{TRIBES}(x)]\leq1/4$.
This would allow us to conclude, since by averaging there must a polynomial
$p\in D$ such that $\P_{x}[p(x)\neq\mathrm{TRIBES}(x)]\leq1/4$. 

To construct $D$, first note the $n/w$ $AND$ terms can be computed
by degree $w$ monomials $m_{1}(x),\dots,m_{n/w}(x)$. To sample $q\sim D$,
we uniformly sample $T_{1},T_{2}\subseteq[n/w]$ and set
\[
q(x):=1-(1-\bigoplus_{i\in T_{1}}m_{i}(x))\wedge(1-\bigoplus_{i\in T_{2}}m_{i}(x)).
\]

Since $T_{1},T_{2}$ are chosen uniformly, for any $x$ such that
$(m_{1}(x),\dots m_{n/w}(x))\neq0$ we have $\P_{q\sim D}[q(x)=0]=1/4$
. And for any $x$ such that $(m_{1}(x),\dots m_{n/w}(x))=0$ we have
$\P_{q\sim D}[q(x)=1]=0$. Together this implies for any $x$, $\P_{q\sim D}[q(x)\neq\mathrm{TRIBES}(x)]\leq1/4$. 
\end{proof}
We are now ready to prove the main result.
\begin{proof}[Proof of Theorem \ref{thm:counterex-CHLLZ}]
 First we note that if $\Delta_{S}(e(p))\leq\e$ then by Markov's
inequality

\begin{equation}
\P_{x^{\overline{S}}}\left[\left|\E_{x^{S}}[e(p(x))]-\E[e(p)]\right|>\e^{1/4}\right]\leq\e^{1/2}.\label{eq:low-degree-property}
\end{equation}
Then, using $T(x)$ to denote $\mathrm{TRIBES}(x)$ for brevity, we
can write
\begin{align*}
\E_{x}\left[e(T(x))e(p(x))\right] & =\E_{x}\left[e(T(x))\cdot(e(p(x))-\E[e(p)])\right]+\E[e(T)]\E[e(p)]\\
 & \leq\E_{x}\left[e(T(x))\cdot(e(p(x))-\E[e(p)])\right]+O(\log n)/n
\end{align*}
where the $\leq$ follows since $|\E[e(T)]|\leq O(\log n)/n$ by the
definition of $\mathrm{TRIBES}$. 

After a uniform assignment to $x^{\overline{S}}$, let $E_{1}$ denote
the event $\left|\E_{x^{S}}[e(p(x))]-\E[e(p)]\right|\leq\e^{1/4}$
and let $E_{2}$ denote the event that $\mathrm{TRIBES}(x)$ is fixed.
Then we have
\begin{align*}
\E_{x}\left[e(T(x))\cdot(e(p(x))-\E[e(p)])\right] & \leq\E_{x^{\overline{S}}}\bigg[\bigg|\E_{x^{S}}\big[e(T(x))\cdot(e(p(x))-\E[e(p)])\big]\bigg|\bigg]\\
 & \leq\E_{x^{\overline{S}}}\bigg[\bigg|\E_{x^{S}}\big[e(T(x))\cdot(e(p(x))-\E[e(p)])\big]\bigg||E_{1}E_{2}\bigg]+ \\ &  \qquad  \P[\neg E_{1}]+\P[\neg E_{2}]\\
 & \leq\e^{1/4}+\e^{1/2}+1/2+o(1).
\end{align*}
For the last inequality, note that $\E_{x^{S}}[e(T(x))\cdot(e(p(x))-\E[e(p)])|E_{1}E_{2}]=\E_{x^{S}}[e(p(x))-\E[e(p)]|E_{1}]$
since $\mathrm{TRIBES}(x)$ is fixed conditioned on $E_{2}$. We bound
$\P[\neg E_{1}]$ by (\ref{eq:low-degree-property}) and $\P[\neg E_{2}]$
by Lemma \ref{lem:tribes-resilient}. Setting $\e$ to a small enough
constant contradicts Lemma \ref{cor:lem-degree-resilient} and concludes
the proof of Theorem \ref{thm:counterex-CHLLZ}.
\end{proof}

\section{Derivatives \label{sec:Derivatives}}

In this section we rewrite $C_{\phi}(p)^{2}$ in terms of the correlation
of the derivatives of $p$ with $Mod_{\phi}$, and use this viewpoint
to derive several basic facts which will be used later. Fix any $\phi\in[0,2\pi]$,
let $\omega=e^{\phi\sqrt{-1}}$, and from here on we let $\sigma:=\sin\phi,\gamma:=\cos\phi$.

We begin by using the fact that $|z|^{2}=z\overline{z}$ for any complex
number, where $\overline{z}$ is the complex conjugate, to rewrite
the correlation square $C_{\phi}^{2}(p)$ as
\[
\E_{x}(-1)^{p(x)}\omega^{\sum_{i}x_{i}}\cdot\overline{\E_{y}(-1)^{p(y)}\omega^{\sum_{i}y_{i}}}.
\]
Replacing $y$ with $x\oplus y$ and noting that $\overline{(-1)^{p(y)}\omega^{\sum_{i}y_{i}}}=(-1)^{p(y)}\omega^{-\sum_{i}y_{i}}$
we can rewrite the correlation square with the following expression:

\[
\E_{y}\E_{x}(-1)^{p(x)+p(x\oplus y)}\omega^{\sum_{i}x_{i}-\sum_{i}(x_{i}\oplus y_{i})}.
\]
The inner expectation over $x$ plays an important role and so we
introduce a definition.
\begin{defn}
\label{def:The-contribution-of} The \emph{contribution }of polynomial
$p$ in the \emph{direction $y$, }or the \emph{$y$-contribution}
of $p$, is $c_{y}(p):=\E_{x}(-1)^{p(x)+p(x\oplus y)}\omega^{\sum_{i}x_{i}-\sum_{i}(x_{i}\oplus y_{i})}$.
\end{defn}

Note $c_{y}(p)$ is always defined with respect to an angle $\phi$,
which will always be clear from context. Repeating what was said above,
\[
C_{\phi}(p)^{2}=\E_{y}c_{y}(p).
\]
The polynomial $p(x)+p(x\oplus y)$ that appears in $c_{y}(p)$ is
the \emph{derivative} of $p$ in \emph{direction} $y$, denoted $p_{y}$.
When $p$ is quadratic, this derivative is linear. Hence, $p_{y}(x)=\sum_{i\le n}p_{y,i}x_{i}+p_{y,0}$
where for every $y$, $p_{y,i}\in\zo$ are is the coefficient of $x_{i}$,
and $p_{y,0}$ is the constant. 

Because $p_{y}(x)$ is linear, for fixed $y$ the expectation over
$x$ is actually the expectation of \emph{independent} functions of
the $x_{i}$ and so the $y$-contribution can be written as
\[
(-1)^{p_{y,0}}\prod_{i=1}^{n}\E_{x_{i}}(-1)^{p_{y,i}x_{i}}\omega^{x_{i}-(x_{i}\oplus y_{i})}.
\]
Each of the expectations $\E_{x_{i}}(-1)^{p_{y,i}x_{i}}\omega^{x_{i}-(x_{i}\oplus y_{i})}$
above takes one of four different values, depending on the four possibilities
for $p_{y,i}$ and $y_{i}$. These values play a crucial role in this
paper and we present them next. Note that if $y_{i}=0$ then $x_{i}-(x_{i}\oplus y_{i})=0$
and so the $\omega$ factor disappears.
\begin{prop}
\label{prop:table}We have the following four possible values for
$\E_{x_{i}}(-1)^{p_{y,i}x_{i}}\omega^{x_{i}-(x_{i}\oplus y_{i})}$:
\end{prop}

\begin{tabular}{|c|c|l|}
\hline 
$p_{y,i}$ & $y_{i}$ & $\E_{x_{i}}(-1)^{p_{y,i}x_{i}}\omega^{x_{i}-(x_{i}\oplus y_{i})}$\tabularnewline
\hline 
\hline 
0 & 0 & $=1$\tabularnewline
\hline 
1 & 0 & $=\E_{x_{i}}(-1)^{x_{i}}=0$\tabularnewline
\hline 
0 & 1 & = $\E_{x_{i}}\omega^{x_{i}-(x_{i}\oplus1)}=\frac{1}{2}\left(\omega^{-1}+\omega\right)=\gamma$\tabularnewline
\hline 
1 & 1 & = $\E_{x_{i}}(-1)^{x_{i}}\omega^{x_{i}-(x_{i}\oplus1)}=\frac{1}{2}\left(\omega^{-1}-\omega\right)=\text{\ensuremath{-\sqrt{-1}\cdot\sigma}}$\tabularnewline
\hline 
\end{tabular}

\paragraph{Restricting to $\phi\in[0,\pi/2]$.}

We now justify our previous assertion that we can restrict our attention
to angles $\phi\in[0,\pi/2]$. First, if $\phi\in[\pi/2,3\pi/2]$
then we can sum $e^{1}$ to $p$. Then $C_{\phi}(p+e^{1})=C_{\pi+\phi}(p)$
and $\pi+\phi\in[-\pi/2,\pi/2]$. Next, if $\phi\in[-\pi/2,0]$ then
$C_{\phi}(p)=C_{-\phi}(p)$ and now $\phi\in[0,\pi/2]$.
\begin{defn}
We denote the Hamming weight of $x\in\zo^{n}$ by $w(x)$.
\end{defn}

Looking at the table above we can obtain the following bound on $c_{y}(p)$
in terms of the weight of the derivative.
\begin{claim}[Weight bound on contribution]
\label{claim:trivial-maximum} For any $y\in\zo^{n}$ and any $\phi$
we have $|c_{y}(p)|\leq\max\{\sigma,\gamma\}{}^{w(y)}.$
\end{claim}

We conclude this section by giving a quick illustration of how this
framework can be used to compute the maximum correlation for $\phi\in[0,\pi/4]$.
Note that Theorem \ref{thm:main-general} proves a stronger result,
showing that non-symmetric polynomials have correlation a constant-factor
smaller than optimal. For such $\phi$ we are going to show that the
constant polynomial, which is symmetric, maximizes $C_{\phi}$. By
Example \ref{exa:cor},
\[
C_{\phi}^{2}(0)=2^{-n}\left(1+\gamma\right)^{n}.
\]
We show this is an upper bound for any quadratic polynomial $p$.
We have
\[
C_{\phi}^{2}(p)\le\E_{y}|c_{y}(p)|,
\]
where $c_{y}$ is as in Definition \ref{def:The-contribution-of}.
By Claim \ref{claim:trivial-maximum}, since $\gamma>\sigma$, we
have 
\[
|c_{y}(p)|\le\gamma^{w(y)}.
\]

Hence,
\[
C_{\phi}^{2}(p)\le2^{-n}\sum_{i=0}^{n}\binom{n}{i}\gamma^{i}=2^{-n}(1+\gamma)^{n},
\]
by the binomial theorem.

\section{Correlation of symmetric polynomials \label{sec:symmetry}}

We use the information from Section \ref{sec:Derivatives} to compute
the maximal correlation of symmetric quadratic polynomials, and note
an important \textquotedblleft no-cancellation\textquotedblright{}
property which will guide the rest of the proof.

We first apply Proposition \ref{prop:table} to determine the \emph{contributions
of symmetric polynomials.} The derivatives of $e^{1}$ are simply
the constant term $e_{y,0}^{1}=\sum_{i}y_{i}$. We now analyze the
derivatives of $e^{2}$. The coefficient $e_{y,i}^{2}$ for $i\ge1$
equals to $\sum_{j\ne i}y_{j}$ and the constant term $e_{y,0}^{2}$
equals $\sum_{i<j}y_{i}y_{j}$. Combining this information with the
above we can characterize the $y$-contributions of symmetric polynomials.
\begin{lem}[Contributions of symmetric polynomials]
\label{lem:=00005BContributions-of-symmetric} For any $\phi\in[0,\pi/2]$
and any $y\in\zo^{n}$ we have:
\begin{enumerate}
\item If $w(y)$ is even then $c_{y}(s)=\sigma{}^{w(y)}$ for either $s=e^{2}+e^{1}$
or $s=e^{2}$.
\item If $w(y)$ is odd and $w(y)<n$ then $c_{y}(s)=0$ for either $s=e^{2}+e^{1}$
or $s=e^{2}$.
\item If $w(y)=n$ and $n\equiv1\bmod4$ then $c_{y}(s)=+\gamma^{n}$ for
$s=e^{2}$ and $c_{y}(s)=-\gamma{}^{n}$ for $s=e^{1}+e^{2}$.
\item If $w(y)=n$ and $n\equiv3\bmod4$ then $c_{y}(s)=-\gamma^{n}$ for
$s=e^{2}$ and $c_{y}(s)=+\gamma$ for $s=e^{1}+e^{2}$.
\end{enumerate}
\end{lem}

\begin{proof}
Refer to Proposition \ref{prop:table}. 

If $w(y)$ is even, the expectations over $x_{i}$ with $y_{i}=0$
contribute $1$ since the corresponding coefficient $s_{y,i}$ (the
coefficient of $x_{i}$ in the derivative polynomial $s_{y}$) is
$0$. This corresponds to the first row of Proposition \ref{prop:table}.
The other expectations contribute $(-\sqrt{-1})\sigma$. This corresponds
to the last row of Proposition \ref{prop:table}. In addition, we
have the constant term. For $e^{2}$ this term is $(-1)^{\binom{w(y)}{2}}=(-1)^{w(y)^{2}/2-w(y)/2}=(-1)^{-w(y)/2}$
using that $w(y)$ is even. For $e^{2}+e^{1}$ the constant term is
$(-1)^{\binom{w(y)}{2}+w(y)}$ which again equals $(-1)^{-w(y)/2}$
because $w(y)$ is even. Hence the $y$-contribution equals
\[
(-1)^{-w(y)/2}\cdot((-\sqrt{-1})\sigma)^{w(y)}=\sigma{}^{w(y)}
\]
where the last equality follows again because $w(y)$ is even.

If $w(y)$ is odd and less than $n$ then some $y_{i}$ is zero. The
corresponding $s_{y,i}$ equals $w(y)$, which is odd. So the contribution
is zero, by the second row of Proposition \ref{prop:table}.

Finally, consider $w(y)=n$ when $n$ is odd. Note that $s_{y,i}=n-1$
which is even. By the third row of Proposition \ref{prop:table},
the expectation of $x$ is $\gamma{}^{n}$ times the constant term.
For $s=e^{2}$ the constant term is $(-1)^{\binom{n}{2}}=(-1)^{n(n-1)/2}$
which is $1$ if $n\equiv1\mod4$ and $-1$ otherwise. For $s=e^{2}+e^{1}$
the constant term is $(-1)^{\binom{n}{2}+n}=(-1)^{n(n-1)/2+1}$ which
is $-1$ if $n\equiv1\mod4$ and $1$ otherwise.
\end{proof}
Lemma \ref{lem:=00005BContributions-of-symmetric} yields an expression
for the maximum $C_{\phi}(s)$ attained by symmetric quadratic polynomials
$s$. It is best to express this correlation using the quantity $v_{\phi}$
that we redefine in a way that is more convenient for the main proof.
\begin{defn}[$E,O,v$]
 Let $E\subseteq\zo^{n}$ be the set of $n$-bit strings of even
Hamming weight, and let $O$ be the set of strings of odd weight.
Define
\[
v_{\phi}:=2^{-n}\sum_{y\in E}\sigma{}^{w(y)}.
\]
The equivalence between this definition and the one in the introduction
is given by the following claim, which we will use often.
\end{defn}

\begin{claim}[Odd-even sum]
 \label{claim:=00005BOdd-even-sum-claim=00005D}For any number $d$
we have:

$\sum_{y:y\in E}d^{w(y)}=\sum_{y}d^{w(y)}(1+(-1)^{w(y)})/2=\frac{(1+d)^{n}+(1-d)^{n}}{2}$,

$\sum_{y:y\in O}d^{w(y)}=\sum_{y}d^{w(y)}(1-(-1)^{w(y)})/2=\frac{(1+d)^{n}-(1-d)^{n}}{2}.$
\end{claim}

\begin{proof}
In each line, the second equality follows from the binomial theorem.
\end{proof}
For example, $v_{2\pi/3}=\Theta((1+\sqrt{3}/2)/2)^{n}$, where $(1+\sqrt{3}/2)/2=0.933\ldots$.
We now give the maximal correlation of a symmetric quadratic polynomial.
\begin{cor}
\label{cor:correlation-of-symmetric} Fix $\phi\in[\pi/2,\pi/4)$
and let $C_{\phi}^{*}$ be the maximum $C_{\phi}$ attained by a symmetric
quadratic polynomial on $n$ bits for large enough $n$. We have:

$C_{\phi}^{*}=\sqrt{v_{\phi}}$ if $n$ is even. This is attained
by both $e^{2}$ and $e^{1}+e^{2}$.

$C_{\phi}^{*}=\sqrt{v_{\phi}+1/4^{n}}$ if $n$ is odd. This is attained
by $e^{2}$ if $n\equiv1\mod4$ and by $e^{1}+e^{2}$ if $n\equiv3\mod4$.
\end{cor}

\begin{proof}
By Example \ref{exa:cor}, $C_{\phi}(e^{1})<C_{\phi}(0)=\left(\frac{1+\gamma}{2}\right)^{n/2}$.
By the definition of $v_{\phi}$, $\sqrt{v_{\phi}}\ge\Omega\left(\left(\frac{1+\sigma}{2}\right)^{n/2}\right)$
which is greater for $n$ large enough since $\sigma>\gamma$ when
$\phi\in[\pi/2,\pi/4)$. The proof now follows from Lemma \ref{lem:=00005BContributions-of-symmetric}.
\end{proof}

\subparagraph{(No) cancellations}

Note an interesting fact holds for the symmetric polynomial that maximizes
$C_{\phi}$: \emph{the $y$-contributions are always real and non-negative},
for any $y$. This is not true in general. For a simple example, take
$p=e^{2},n=3\bmod4,$ and $w(y)=n.$ Then $c_{y}(p)$ is negative.
This leads to cancellations in the correlation. However, for the symmetric
polynomial that maximizes correlation, the inner expectation is always
non-negative and there are \emph{no cancellations}. 

This fact shows that for the symmetric polynomials $p$ that maximize
correlation, the correlation square $C_{\phi}^{2}(p)$ can be equivalently
written as
\[
\E_{y}|c_{y}(p)|;
\]
that is, we can take absolute values of the contributions ``for free''.
Note that by the triangle inequality, for \emph{any }polynomial $p$
the above expression is an \emph{upper bound} on the correlation.
We used this when showing the constant polynomial maximizes $C_{\phi}$
for $\phi\in[0,\pi/4].$ For the symmetric polynomials that maximize
correlation, it turns out that this bound can be attained.

In the proof of Theorem \ref{thm:main-general} we shall mostly be
working with this quantity, which does not depend on the linear part
of $p$. This is because the derivative of a linear polynomial is
a constant depending only on $y$, which disappears when taking absolute
values. Hence we can assume that $p$ does not contain linear terms.

\section{Proof of Theorem \ref{thm:main-general} \label{sec:proof-main}}

The next two results are needed to prove the first, main item of Theorem
\ref{thm:main-general}. First we deal with polynomials that are missing
at least one degree two monomial.
\begin{thm}
\label{thm:c(p)-upper-bound} Let $\phi\in(\pi/4,\pi/2]$ and $p$
be a quadratic polynomial that is not equal to $e^{2}+\ell$ for some
linear polynomial $\ell$. Then $\E_{y}\left|c_{y}(p)\right|\le(1-\Omega(\sigma-\gamma))v_{\phi}$.
\end{thm}

Next we deal with non-symmetric polynomials that possess all degree
two monomials. Note we use the quantity $\E_{y}c_{y}(p)$ instead.
\begin{lem}
\label{lem:e2-upper-bound}Let $\phi\in(\pi/4,\pi/2]$ and $p$ be
a polynomial that is equal to $e^{2}+\ell$ where $\ell$ is a linear
polynomial not equal to a constant or $e^{1}$. Then $\E_{y}c_{y}(p)\le(1-\Omega(1))v_{\phi}$.
\end{lem}

Assuming these are true, we prove the first item of Theorem \ref{thm:main-general}.
\begin{proof}[Proof of Theorem \ref{thm:main-general} Item 1]
 Follows from Corollary \ref{cor:correlation-of-symmetric}, Theorem
\ref{thm:c(p)-upper-bound}, and Lemma \ref{lem:e2-upper-bound}.
\end{proof}
We next give similar results that are needed to prove the second item
of Theorem \ref{thm:main-general}.
\begin{lem}
\label{thm:c(p)-upper-bound-easy} Let $\phi\in[0,\pi/4]$ and $p$
be a quadratic polynomial that is not linear. Then $\E_{y}\left|c_{y}(p)\right|\le(1-\Omega(1))\left(\frac{1+\gamma}{2}\right)^{n}$.
\end{lem}

\begin{lem}
\label{lem:e2-upper-bound-easy}Let $\phi\in[0,\pi/4]$ and $p$ be
a linear polynomial that is not equal to the constant polynomial.
Then $\E_{y}c_{y}(p)\le(1-\Omega(1))\left(\frac{1+\gamma}{2}\right)^{n}$.
\end{lem}

\begin{proof}[Proof of Theorem \ref{thm:main-general} Item 2]
 Follows from Lemma \ref{thm:c(p)-upper-bound-easy}, Lemma \ref{lem:e2-upper-bound-easy},
and Example \ref{exa:cor} which says $C_{\phi}^{2}(0)=\left(\frac{1+\gamma}{2}\right)^{n}$.
\end{proof}

\subsection{Proof of Theorem \ref{thm:c(p)-upper-bound} \label{sec:Proof-of-Theorem}}

Our proof strategy is to \emph{slowly restrict the direction }$y,$
to try to connect the corresponding contributions with the target
value $v_{\phi}$.
\begin{defn}
A \emph{restriction} $r$ is an element of $\{0,1,*\}^{n}$. The \emph{weight}
$w(r)$ of $r$ is the number of ones, and $S(r)$ is the number of
stars. We also view $r$ as a function $r:\zo^{S(r)}\to\zo^{n}$ mapping
assignments to stars to $n$-bit strings, and we write $ry$ for $r(y)$.
For a restriction $r$ we call $x_{i}$ a $b\in\{0,1,*\}$ variable
if the $i$th bit of $r$ is $b$.
\end{defn}

We emphasize that $r$ restricts the space of directions $y$, not
$x$. So for example if $x_{i}$ is a $0$ variable then the corresponding
directional bit $y_{i}$ has been restricted to $0$ -- but $x_{i}$
is never restricted. We next introduce restricted versions of the
quantities in Theorem \ref{thm:c(p)-upper-bound}.
\begin{defn}[$c(p,r)$ and $v_{\phi}(r)$]
 Let $r$ be a restriction. For a polynomial $p$ we define 
\[
c(p,r):=\E_{y\in\zo^{S(r)}}|c_{ry}(p)|.
\]
Note that $c(p,r)$ is defined with respect to the angle $\phi$ since
$c_{ry}(p)$ is. We also define 
\[
v_{\phi}(r):=2^{-S(r)}\sum_{y\in\zo^{S(r)}:ry\in E}\sigma{}^{w(ry)},
\]
where we sum over all derivatives $ry$ of even weight. 
\end{defn}

For any $r\in\zo^{n}$ we have $c(p,r)=|c_{r}(p)|$. Also,
\begin{align*}
\E_{y}|c_{y}(p)| & =c(p,*^{n}),\\
v_{\phi} & =v_{\phi}(*^{n}).
\end{align*}
Using the above notation our goal is to show that 
\[
c(p,*^{n})\leq(1-\Omega(\sigma-\gamma))v_{\phi}.
\]

\paragraph{Polynomials as graphs}

We associate to a quadratic polynomial $p$ the \emph{graph }over
the variables where $x_{i}$ and $x_{j}$ are connected iff monomial
$x_{i}x_{j}$ is present in $p$. Note this graph only depends on
the monomials of degree $2$ of $p$. The \emph{degree} of a variable
shall refer to the degree as a node in this graph. We shall also talk
of variables being connected, etc.
\begin{example}
Let $n=3,r=(1*0)\in\{0,1,*\}^{3},p=x_{1}x_{2}+x_{2}x_{3}$. The $*$
variable $x_{2}$ is connected to the $1$ variable $x_{1}$ and to
the $0$ variable $x_{3}$.
\end{example}

We now proceed with the proof of Theorem \ref{thm:c(p)-upper-bound}.
In all upcoming statements, $p$ is an arbitrary quadratic polynomial
on $n$ variables, $\phi\in(\pi/4,\pi/2],$ and we set $n$ and a
parameter $t$ large enough so that both $t$ and $n/t$ are large
enough depending on $\phi$. The minimal $n$ for which our proof
of Theorem \ref{thm:main-general} holds increases as $\phi$ approaches
$\pi/4$ (where $\sigma$ approaches $\gamma).$ 

We next state several lemmas and prove Theorem \ref{thm:c(p)-upper-bound}
assuming them. The first two lemmas show that $c(p,r)\le v_{\phi}(r)$
under various conditions on $p$ and $r$.
\begin{lem}
\label{lem:odd-even-general} Let $r\in\{0,1,*\}^{n}$ be a restriction.
Suppose there exists a 0 variable that is connected to an odd number
of 1 variables. Then $c(p,r)\leq v_{\phi}(r).$
\end{lem}

\begin{lem}
\label{lem:gap-general}Let $r\in\{0,1,*\}^{n}$ be a restriction.
Suppose there exists a 0 variable that is connected to an even number
of 1 variables and at least $t$ $*$ variables. Then $c(p,r)\leq v_{\phi}(r).$
\end{lem}

The next lemma shows that if $p$ is missing a degree two monomial
then $v_{\phi}(0*^{n-1})$ gains an advantage over $c(p,0*^{n-1})$.
It can be considered a strengthening of Lemma \ref{lem:gap-general}
under an additional constraint.
\begin{lem}[Buffer]
\label{lem:buffer-general} Let $r=0*^{n-1}$. Suppose the 0 variable
is connected to at least $t$ $*$ variables and at most $n-2$ $*$
variables. Then $c(p,r)\leq v_{\phi}(r)-\left(\frac{\sigma-\gamma}{16}\right)v_{\phi}.$
\end{lem}

We shall use the above lemmas to slowly restrict directions, beginning
with Lemma \ref{lem:buffer-general} and then iteratively applying
either Lemma \ref{lem:odd-even-general} or Lemma \ref{lem:gap-general}.
This process stops when we cannot find variables that satisfy the
hypothesis of either Lemma \ref{lem:odd-even-general} or Lemma \ref{lem:gap-general}. 

When this happens, we consider two cases based on the number of variables
restricted. In the first case, when the number is large, we give an
upper bound on $c(p,r)$. This suffices because of the buffer afforded
to us by Lemma \ref{lem:buffer-general}.
\begin{lem}[Opened majority]
\label{lemma:done-if-most-open}Let $r=1^{j}*^{n-j}$ for some $j\geq n/2$.
Then $c(p,r)<2^{j}\left(\frac{\sigma-\gamma}{1000}\right)v_{\phi}.$
\end{lem}

In the second case, when the number of restricted variables is small,
the polynomial has structure that we can utilize to again show $c(p,r)\le v_{\phi}(r)$.
Specifically, in the graph of the polynomial many variables have small
degree. 
\begin{lem}[Low degree loses]
\label{lemma:low-degree-loses} Let $r=1^{j}*^{n-j}$ for some $j<n/2$.
Suppose every $*$ variable is connected to at most $t$ other $*$
variables. Then $c(p,r)\leq v_{\phi}(r).$
\end{lem}

We will need the following variant of Lemma \ref{lemma:low-degree-loses}
for an edge case in the main proof.
\begin{lem}
\label{lemma:low-degree-loses-all} Let $r=*^{n}$. Suppose there
are at least $n-t$ variables connected to at most $t$ other variables.
Then $c(p,r)\leq(1-(\sigma-\gamma))v_{\phi}.$ 
\end{lem}

Assuming these lemmas we can prove Theorem \ref{thm:c(p)-upper-bound}.
\begin{proof}[Proof of Theorem \ref{thm:c(p)-upper-bound}]
 We consider two cases based on the existence of a variable of certain
degree in the graph of $p$. In the first case, when $p$ is a `typical'
polynomial, we suppose the existence of a variable with degree in
$[t,n-2]$ (corresponding to the hypothesis of Lemma \ref{lem:buffer-general}).
Let us denote this variable $x_{1}$ for ease. We ``open'' the directional
bit corresponding to $x_{1}$. That is, we condition $\E_{y}|c_{y}(p)|$
depending on the value of $y_{1}$:
\[
c(p,*^{n})=\frac{1}{2}\left(c(p,0*^{n-1})+c(p,1*^{n-1})\right).
\]
Correspondingly, it holds that
\[
v_{\phi}(*^{n})=\frac{1}{2}\left(v_{\phi}(0*^{n-1})+v_{\phi}(1*^{n-1})\right).
\]
Then we iteratively open up $*$ variables in the term where the restriction
has no zeroes, as long as we can find a $*$ variable that is connected
to an  number of $1$ variables or that is connected to an even number
of $1$ variables and at least $t$ other $*$ variables. We can write
the terms corresponding to the variables that were opened (up to permutation
of variables):
\[
c(p,*^{n})=\frac{1}{2}c(p,0*^{n-1})+\frac{1}{4}c(p,10*^{n-2})+\dots+\frac{1}{2^{j}}c(p,1^{j}*^{n-j}),
\]
for some $1\leq j\leq n$ depending on $p$. We also write the corresponding
terms for $v_{\phi}$:
\[
v_{\phi}(*^{n})=\frac{1}{2}v_{\phi}(0*^{n-1})+\frac{1}{4}v_{\phi}(10*^{n-2})+\dots+\frac{1}{2^{j}}v_{\phi}(1^{j}*^{n-j}).
\]

We compare the terms in the right-hand sides in the two equations
above. For the first term, we have $\frac{1}{2}c(p,0*^{n-1})\leq\frac{1}{2}v_{\phi}(0*^{n-1})-(\frac{\sigma-\gamma}{32})v_{\phi}$
by Lemma \ref{lem:buffer-general}. For all the other terms except
the last one, we have that the $c(p,r)$ terms is at most the corresponding
$v_{\phi}(r)$ term by either Lemma \ref{lem:odd-even-general} or
Lemma \ref{lem:gap-general}. Now we analyze the last terms depending
on the value of $j$. Note that each $*$ variable is connected to
at most $t$ other $*$ variables.

If $1\leq j<n/2$ we apply Lemma \ref{lemma:low-degree-loses} which
says $c(p,1^{j}*^{n-j})\leq v_{\phi}(p,1^{j}*^{n-j})$ and conclude
as $v_{\phi}(*^{n})-c(p,*^{n})\geq\frac{\sigma-\gamma}{32}v_{\phi}$
. 

If $j\ge n/2$ then $\frac{1}{2^{j}}c(p,1^{j}*^{n-j})\leq(\frac{\sigma-\gamma}{1000})v_{\phi}$
by Lemma \ref{lemma:done-if-most-open} and we conclude as $v_{\phi}(*^{n})-c(p,*^{n})\geq\left(\frac{\sigma-\gamma}{32}-\frac{\sigma-\gamma}{1000}\right)v_{\phi}$. 

This finishes the proof of when $p$ has a node with degree in $[t,n-2]$.
For the second case, suppose that every node has degree at most $t-1$
or degree exactly $n-1$. We then claim there are $\leq t-1$ nodes
with degree $n-1$. Supposing this is true we can immediately conclude
by Lemma \ref{lemma:low-degree-loses-all}.

Now we verify the desired claim. Suppose there are $z$ nodes of degree
$n-1$ with $z\ge t$. Each of these nodes is connected to every other
node, so every node in the graph has degree at least $z\ge t$. By
the supposition, every node in the graph has degree $n-1$. This contradicts
the hypothesis that $p\neq e^{2}+\ell$.
\end{proof}
Next we give proofs of the technical lemmas.

\subsubsection{Proof of Lemma \ref{lem:odd-even-general}}

Fix a 0 variable $x_{i}$ that is connected to an odd number of 1
variables. Let $T$ denote the indices of the $*$ variables connected
to $x_{i}$ and let $U$ denote the indices of the remaining $*$
variables. Write $y=(y^{T},y^{U})$ for the corresponding bits of
$y$.

Note that by Proposition \ref{prop:table}, $c_{ry}(p)=0$ if $w(y^{T})$
is even (because the coefficient of $x_{i}$ would be odd). And if
$w(y^{T})$ is odd we apply the upper bound $|c_{ry}(p)|\leq\sigma^{w(ry)}$
from Claim \ref{claim:trivial-maximum}. Combining these two things
yields:
\begin{align*}
c(p,r) & =2^{-S(r)}\sum_{y^{T}\in O,y^{U}}|c_{ry}(p)|\\
 & \leq2^{-S(r)}\sum_{y^{T}\in O,y^{U}}\sigma{}^{w(ry)}.
\end{align*}
Now we compare this value with the expression for $v_{\phi}$. Let
us assume that $w(r)$ is even. Then
\begin{align*}
v_{\phi}(r) & =2^{-S(r)}\sum_{y\in E}\sigma{}^{w(ry)}.
\end{align*}
Hence to prove $c(p,r)\le v_{\phi}(r)$ it suffices to show
\[
\sum_{y^{T}\in O,y^{U}}\sigma{}^{w(y)}\le\sum_{y\in E}\sigma{}^{w(y)}.
\]
Note in the above two expressions we can assume $|T|>0$ since otherwise
the left hand-side will be 0 and we would be immediately done. Then
by conditioning on the parity of $y^{U}$ in each side it suffices
to show
\[
\sum_{y^{T}\in O,y^{U}\in E}\sigma{}^{w(y)}+\sum_{y^{T}\in O,y^{U}\in O}\sigma{}^{w(y)}\le\sum_{y^{T}\in E,y^{U}\in E}\sigma{}^{w(y)}+\sum_{y^{T}\in O,y^{U}\in O}\sigma{}^{w(y)}.
\]
The second sum in each side is the same, and the first sum in the
right-hand side is bigger than the first sum in the left-hand side
by Claim \ref{claim:=00005BOdd-even-sum-claim=00005D}. This concludes
the case of when $w(r)$ is even. 

When $w(r)$ is odd 
\[
v_{\phi}(r)=2^{-S(r)}\sum_{y\in O}\sigma{}^{w(ry)}.
\]
Then it suffices to show 
\[
\sum_{y^{T}\in O,y^{U}\in E}\sigma{}^{w(y)}+\sum_{y^{T}\in O,y^{U}\in O}\sigma{}^{w(y)}\le\sum_{y^{T}\in E,y^{U}\in O}\sigma{}^{w(y)}+\sum_{y^{T}\in O,y^{U}\in E}\sigma{}^{w(y)}.
\]
The inequality holds again by Claim \ref{claim:=00005BOdd-even-sum-claim=00005D}.

\subsubsection{Proof of Lemma \ref{lem:gap-general}}

The high-level approach is similar to the proof of Lemma \ref{lem:odd-even-general},
but we utilize the following improvement of Claim \ref{claim:trivial-maximum}
when the weight of the derivative is odd. The improvement comes from
the handshaking lemma\emph{.}
\begin{claim}
\label{claim:contribution-is-even} Let $y\in\zo^{n}$. Then $|c_{y}(p)|$
is either 0 or $\sigma{}^{e}\gamma{}^{w(y)-e}$, where $e$ is an
even integer and $0\leq e\leq w(y)$.
\end{claim}

\begin{proof}
Consider the graph $G$ with $w(y)$ nodes which are the 1 variables
and the edges represent monomials. Let $S,T$ be the nodes in $G$
that have odd, even degree respectively. Note that nodes in $S$ contribute
a $\sigma$ factor, while the nodes in $T$ contribute a $\gamma$
factor. The remaining $n-w(y)$ 0 variables not in $G$ contribute
either $1$ or $0$.

So to finish the proof it suffices to show that $|S|$ must be even.
The sum of all the degrees in $G$ is $|S|\cdot odd+(|V|-|S|)\cdot even=|S|\cdot odd+even$.
In any graph, the sum of degrees is even, hence $|S|$ is always even.
\end{proof}
To prove Lemma \ref{lem:gap-general} we exploit that if $w(ry)$
is odd then the exponent of the $\sigma$ factor is $<w(ry).$ Fix
the 0 variable $x_{i}$ that is connected to an even number of 1 variables
and to at least $t$ $*$ variables. Let $T$, $U$ denote the same
as in the previous proof. The $ry$ contribution is zero if $w(y^{T})$
is odd (because the coefficient of $x_{i}$ in the $ry$ derivative
would be $even+odd=odd$). So then
\begin{align*}
c(p,r) & =2^{-S(r)}\sum_{y^{T}\in E,y^{U}}|c_{ry}(p)|\\
 & =2^{-S(r)}(\sum_{y^{T}\in E,y^{U}\in E}|c_{ry}(p)|+\sum_{y^{T}\in E,y^{U}\in O}|c_{ry}(p)|).
\end{align*}

Suppose that $w(r)$ is even. For the first term, where $y^{T}\in E,y^{U}\in E$,
we use Claim \ref{claim:trivial-maximum}. For the second term, where
$y^{T}\in E,y^{U}\in O$, $w(ry)=even+even+odd=odd$. By Claim \ref{claim:contribution-is-even},
the max contribution of $ry$ in the second term is $\leq\sigma{}^{w(ry)-1}\gamma$.
So we can bound

\[
c(p,r)\leq2^{-S(r)}(\sum_{y^{T}\in E,y^{U}\in E}\sigma{}^{w(ry)}+\frac{\gamma}{\sigma}\sum_{y^{T}\in E,y^{U}\in O}\sigma{}^{w(ry)}).
\]
We compare this to 
\begin{align*}
v_{\phi}(r) & =2^{-S(r)}\sum_{y\in E}\sigma{}^{w(ry)}\\
 & =2^{-S(r)}(\sum_{y^{T}\in E,y^{U}\in E}\sigma{}^{w(ry)}+\sum_{y^{T}\in O,y^{U}\in O}\sigma{}^{w(ry)}).
\end{align*}
The sums over $y^{T}\in E,y^{U}\in E$ are the same. Hence to show
$c(p,r)\le v_{\phi}(r)$ it suffices to show
\begin{align*}
\frac{\gamma}{\sigma}\sum_{y^{T}\in E,y^{U}\in O}\sigma{}^{w(y)} & \le\sum_{y^{T}\in O,y^{U}\in O}\sigma{}^{w(y)}\\
\iff\frac{\gamma}{\sigma}\sum_{y^{T}\in E}\sigma{}^{w(y^{T})} & \le\sum_{y^{T}\in O}\sigma{}^{w(y^{T})}\\
\iff\left(\sigma/\gamma+1\right)(1-\sigma)^{|T|} & \leq\left(\sigma/\gamma-1\right)(1+\sigma)^{|T|}\\
\iff\frac{\sigma+\gamma}{\sigma-\gamma} & \leq\left(\frac{1+\sigma}{1-\sigma}\right)^{|T|}.
\end{align*}
The second to last $\iff$ follows by applying Claim \ref{claim:=00005BOdd-even-sum-claim=00005D}
and rearranging. The last inequality holds for $t$ large enough,
since $|T|\geq t$ and the left hand term will be some fixed positive
number since $\phi\in(\pi/4,\pi/2]$. This concludes the $w(r)$ even
case.

Now suppose $w(r)$ is odd. Proceeding similarly as before, we have
\[
c(p,r)\leq2^{-S(r)}(\frac{\gamma}{\sigma}\sum_{y^{T}\in E,y^{U}\in E}\sigma{}^{w(ry)}+\sum_{y^{T}\in E,y^{U}\in O}\sigma{}^{w(ry)}).
\]
Which we need to compare with 
\begin{align*}
v_{\phi}(r) & =2^{-S(r)}\sum_{y\in O}\sigma{}^{w(ry)}\\
 & =2^{-S(r)}(\sum_{y^{T}\in E,y^{U}\in O}\sigma{}^{w(ry)}+\sum_{y^{T}\in O,y^{U}\in E}\sigma{}^{w(ry)}).
\end{align*}
Now the sums over $y^{T}\in E,y^{U}\in O$ are the same. So then it
suffices to show
\begin{align*}
\frac{\gamma}{\sigma}\sum_{y^{T}\in E}\sigma{}^{w(y^{T})} & \leq\sum_{y^{T}\in O}\sigma{}^{w(y^{T})}
\end{align*}
which we have already verified.

\subsubsection{Proof of Lemma \ref{lem:buffer-general}}

The proof starts identically as the proof of Lemma \ref{lem:gap-general},
but then we strengthen the analysis to give a strict inequality. Let
$T$ denote the set of $*$ variables connected to $x_{1}$, and let
$U$ denote the $*$ variables not connected to $x_{1}$. We have
$|T|+|U|=n-1$ and by hypothesis $t\leq|T|\leq n-2$. We remark the
strengthened analysis only works because of the condition $|T|\leq n-2$. 

We have the following derivation, where the first inequality follows
from the same steps as in $w(r)$ even case of the previous proof.
Let $a=1+\sigma,b=1-\sigma$, and $\delta=\gamma/\sigma$.
\begin{align*}
2^{n-1}\left(v_{\phi}(0*^{n-1})-c(p,0*^{n-1})\right) & \geq\sum_{y^{T}\in O,y^{U}\in O}\sigma{}^{w(y)}-\frac{\gamma}{\sigma}\sum_{y^{T}\in E,y^{U}\in O}\sigma^{w(y)}.\\
 & =\sum_{y^{U}\in O}\sigma^{w(y^{U})}\cdot(\sum_{y^{T}\in O}\sigma{}^{w(y^{T})}-\delta\sum_{y^{T}\in E}\sigma{}^{w(y^{T})})\\
 & =\frac{a^{|U|}-b^{|U|}}{2}\cdot\frac{(1-\delta)a^{|T|}-(1+\delta)b^{|T|}}{2}\\
 & \geq\frac{a^{|U|}}{4}\cdot\frac{(1-\delta)a^{|T|}}{4}\\
 & =\frac{(1-\delta)a^{n-1}}{16}.
\end{align*}
We elaborate on the last $\geq$. First, note that if $|U|=0$ the
inequality would not be valid since the entire expression would be
equal to 0. Second, we verify that
\begin{align*}
\frac{(1+\delta)b^{|T|}}{2} & \leq\frac{(1-\delta)a^{|T|}}{4}\\
\iff2\cdot\frac{\sigma+\gamma}{\sigma-\gamma} & \leq\left(\frac{1+\sigma}{1-\sigma}\right)^{|T|}.
\end{align*}
The last inequality holds for $t$ large enough, since $|T|\geq t$.
Note this is almost the same inequality that is in the proof of Lemma
\ref{lem:gap-general}. Lastly, we verify that
\begin{align*}
\frac{b^{|U|}}{2} & \leq\frac{a^{|U|}}{4}\\
\Leftarrow2 & \leq\frac{1+\sigma}{1-\sigma}.
\end{align*}
The $\Leftarrow$ holds since $|U|>0$ and the last inequality is
equivalent to $\sigma\ge1/3$ which holds since $\sigma=\sin(\phi)\ge\sin(\pi/4)=1/\sqrt{2}\ge1/3$.

We continue the derivation, applying similar logic:

\begin{align*}
\frac{(1-\delta)a^{n-1}}{16} & \geq\frac{(1-\delta)a^{n-1}+(1-\delta)b^{n-1}}{32}\\
 & \geq\frac{(1-\delta)a^{n}+(1-\delta)b^{n}}{32a}\\
 & =\frac{(1-\delta)}{16a}\cdot2^{n}v_{\phi}.
\end{align*}
Dividing both sides by $2^{n-1}$ we obtain
\begin{align*}
v_{\phi}(0*^{n-1})-c(p,0*^{n-1}) & \geq\frac{(1-\delta)}{8a}\cdot v_{\phi}\\
 & \geq\frac{\sigma-\gamma}{16}\cdot v_{\phi}.
\end{align*}
where the last $\geq$ follows since $a=1+\sigma\leq2$, $(1-\delta)=\frac{\sigma-\gamma}{\sigma}\geq\sigma-\gamma$
because $\sigma\leq1$.

\subsubsection{Proof of Lemma \ref{lemma:done-if-most-open}}

Applying Claim \ref{claim:trivial-maximum} we can say
\begin{align*}
c(p,1^{j}*^{n-j}) & \leq2^{-(n-j)}\sigma^{j}\sum_{y}\sigma{}^{w(y)}\\
 & =2^{-(n-j)}\sigma{}^{j}(1+\sigma)^{n-j}.
\end{align*}
On the other hand,
\begin{align*}
2^{j}v_{\phi}(*^{n}) & \geq2^{-(n-j+1)}(1+\sigma)^{n}.
\end{align*}
So it suffices to show that
\begin{align*}
\frac{\sigma{}^{j}(1+\sigma)^{n-j}}{2^{n-j}} & \leq\frac{\sigma-\gamma}{1000}\frac{(1+\sigma)^{n}}{2^{n-j+1}}\\
\iff\frac{2000}{\sigma-\gamma} & \leq\left(\frac{1+\sigma}{\sigma}\right)^{j},
\end{align*}
where we divided by $\sigma-\gamma>0$. The last inequality holds
for $n$ large enough since $j\geq n/2$ and $\sigma>0$.

\subsubsection{Proof of Lemma \ref{lemma:low-degree-loses}}

Consider the subgraph induced by the $*$ variables. There are $n-j\geq n/2$
nodes in it of degree $\leq t$. By a greedy argument, this implies
an independent set of size $\geq(n-j)/(t+1)\ge n/4t$. Let $T$ denote
the variables in the independent set and let $S$ denote the remaining
$*$ variables. Note $|S|+|T|=n-j$ and the remaining $j$ variables
are 1 variables.

For any fixing $y^{S}$ of $S$, let $p^{T}(y^{S})\in\zo^{|T|}$ denote
the coefficients of the variables in $T$ based on the partial restriction
$1^{j}y^{S}*^{|T|}$. This is a valid definition because $T$ is an
independent set, and so $p^{T}(y^{S})$ is unaffected by any fixing
$y^{T}$ of $T$. By Proposition \ref{prop:table}, if for some fixing
$y^{T}$ there is a variable $x_{j}$ in $T$ such that $p_{j}^{T}(y^{S})=1$
but $y_{j}^{T}=0$ then the contribution is 0. Using also the other
values in the table in Proposition \ref{prop:table}, for any fixed
$y^{S}$ we can let $\psi:=w(p^{T}(y^{S}))$ and bound the contribution
over $y^{T}$ as follows:
\begin{align*}
2^{|T|}c(p,1^{j}y^{S}*^{|T|}) & \le\sigma^{j+w(y^{S})+\psi}\sum_{z\in\zo^{|T|-\psi}}\gamma^{w(z)}\\
 & =\sigma^{j+w(y^{S})+\psi}(1+\gamma)^{|T|-\psi}\\
 & \leq\sigma^{j+w(y^{S})}(1+\gamma)^{|T|}.
\end{align*}
The last $\leq$ follows since $\sigma<1\le1+\gamma$. By summing
over all possible fixings $y^{S}$ and applying the previous bound
we can bound $c(p,1^{j}*^{n-j})$ as follows:
\begin{align*}
2^{n-j}c(p,1^{j}*^{n-j}) & \leq\sigma^{j}(1+\gamma)^{|T|}\sum_{y^{S}}\sigma^{w(y^{S})}\\
 & =\sigma^{j}(1+\gamma)^{|T|}(1+\sigma)^{|S|}\\
 & \leq\sigma^{j}(1+\gamma)^{n/4t}(1+\sigma)^{(n-j)-n/4t}.
\end{align*}
The last $\leq$ holds since $\sigma>\gamma$ and $|T|\geq n/4t$.
On the other hand, 
\begin{align*}
2^{n-j}v_{\phi}(1^{j}*^{n-j}) & =\sum_{y:1^{j}y\in E}\sigma^{j+w(y)}\\
 & \geq\sigma^{j}\frac{(1+\sigma)^{n-j}}{4}.
\end{align*}
So then it suffices to show 
\begin{align*}
\sigma^{j}(1+\gamma)^{n/4t}(1+\sigma)^{(n-j)-n/4t} & <\sigma^{j}\frac{(1+\sigma)^{n-j}}{4}\\
\iff(1+\gamma)^{n/4t} & <\frac{(1+\sigma)^{n/4t}}{4}\\
\iff4 & <\left(\frac{1+\sigma}{1+\gamma}\right)^{n/4t}.
\end{align*}
Since $\sigma>\gamma$ when $\phi\in(\pi/4,\pi/2]$, the last inequality
holds for $n/t$ large enough.

\subsubsection{Proof of Lemma \ref{lemma:low-degree-loses-all} }

The proof is nearly identical to the proof of Lemma \ref{lemma:low-degree-loses}.
The hypothesis implies the existence of an independent set of size
$\geq(n-t)/(t+1)\geq(n-t)/2t$ in the graph consisting of all the
variables. Following the same logic as before, we can upper bound
$c(p,*^{n})$ by
\begin{align*}
2^{n}c(p,*^{n}) & \leq(1+\gamma)^{(n-t)/2t}(1+\sigma)^{n-(n-t)/2t}.
\end{align*}
On the other hand, 
\begin{align*}
2^{n}v_{\phi} & \geq\frac{(1+\sigma)^{n}}{2}.
\end{align*}
Then it suffices to show 
\begin{align*}
(1+\gamma)^{(n-t)/2t}(1+\sigma)^{n-(n-t)/2t} & <(1-(\sigma-\gamma))\frac{(1+\sigma)^{n}}{2}\\
\iff\frac{2}{(1-(\sigma-\gamma))} & <\left(\frac{1+\sigma}{1+\gamma}\right)^{(n-t)/2t}.
\end{align*}
Recall that $n/t$ is arbitrarily large, so $(n-t)/2t$ is also arbitrarily
large and the inequality holds.

\subsection{Proof of Lemma \ref{lem:e2-upper-bound} \label{sec:Proof-of-Lemma-e2}}

We can perform a similar analysis as in the proof of Lemma \ref{lem:=00005BContributions-of-symmetric}.
As before $c_{y}(p)=0$ if $w(y)$ is odd. But now if $w(y)$ even,
letting $T$ denote the set of variables that appear in the linear
polynomial $\ell$, the contribution is
\begin{align*}
c_{y}(p) & =(-1)^{-w(y)/2+w(y^{T})}\cdot((-\sqrt{-1})\sigma)^{w(y)}\\
 & =(-1)^{w(y^{T})}\sigma{}^{w(y)}.
\end{align*}
So a derivative makes a positive contribution if $w(y)$ is even and
$w(y^{T})$ is even, and a negative one if $w(y)$ is even and $w(y^{T})$
is odd. Let $U$ be the complement of $T$. By hypothesis, $1\leq|T|,|U|\leq n-1$.
We can sum over the positive contributions and subtract the negative
ones to get the expression
\begin{align*}
2^{n}\cdot\E_{y}c_{y}(p) & =\sum_{y^{T}\in E,y^{U}\in E}\sigma{}^{w(y)}-\sum_{y^{T}\in O,y^{U}\in O}\sigma{}^{w(y)}.
\end{align*}
On the other hand, 
\[
2^{n}\cdot v_{\phi}=\sum_{y^{T}\in E,y^{U}\in E}\sigma{}^{w(y)}+\sum_{y^{T}\in O,y^{U}\in O}\sigma{}^{w(y)}.
\]
Combining the two expressions and letting $a=(1+\sigma),b=(1-\sigma)$,
we get 
\begin{align*}
2^{n}\left(v_{\phi}-\E_{y}c_{y}(p)\right) & =2\sum_{y^{T}\in O,y^{U}\in O}\sigma{}^{w(y)}\\
 & =\frac{1}{2}\left(a^{|T|}-b^{|T|}\right)\left(a^{|U|}-b^{|U|}\right)\\
 & \ge\frac{1}{2}\frac{a^{|T|}}{2}\frac{a^{|U|}}{2}\\
 & =\frac{a^{n}}{8}.
\end{align*}
The second $=$ follows by Claim \ref{claim:=00005BOdd-even-sum-claim=00005D},
and the $\geq$ after that follows since $1\leq|T|,|U|$ by hypothesis
and $2b<a$.

\subsection{Proof of Lemma \ref{thm:c(p)-upper-bound-easy}}

Since $p$ is not linear there is at least one node with degree $\geq1$
in the polynomial graph. Let us denote this node $x_{1}$ for ease,
and let $T,U$ denote the nodes connected, not connected to $x_{1}$
respectively. We write $y=(y^{T},y^{U})$ for the corresponding bits
of $y$. Just like in the proof of Theorem \ref{thm:c(p)-upper-bound}
we condition on the value of $y_{1}$ to get
\begin{align*}
c(p,*^{n}) & =\frac{1}{2}\left(c(p,0*^{n-1})+c(p,1*^{n-1})\right).
\end{align*}

We bound the second term by applying Claim \ref{claim:trivial-maximum}
which says $c_{ry}(p)\leq\gamma^{w(ry)}$ using that $\phi\in[0,\pi/4]$:
\begin{align*}
2^{n-1}c(p,1*^{n-1}) & \leq\sum_{y\in\zo^{n-1}}\gamma^{1+w(y)}\\
 & =\gamma\left(1+\gamma\right)^{n-1}.
\end{align*}

To deal with the first term, we proceed similarly as we did in the
proof of Lemma \ref{lem:odd-even-general}. Note that $|T|\geq1$,
and if $w(y^{T})$ is odd then $c_{1y}(p)=0$. If $w(y^{T})$ is even
then as before we use the bound $c_{ry}(p)\leq\gamma^{w(ry)}.$ These
two things yield
\begin{align*}
2^{n-1}c(p,0*^{n-1}) & \leq\sum_{y^{T}\in E,y^{U}}\gamma{}^{w(y)}\\
 & =\left(\frac{(1+\gamma)^{|T|}+(1-\gamma)^{|T|}}{2}\right)(1+\gamma)^{|U|}\\
 & \leq3/4(1+\gamma)^{n-1}.
\end{align*}
The last $\leq$ follows as $|T|\geq1$ and $1-\gamma<\frac{1+\gamma}{2}$
when $1/\sqrt{2}\leq\gamma$. Altogether this gives
\[
2^{n}c(p,*^{n})\leq(3/4+\gamma)(1+\gamma)^{n-1}.
\]

So it only remains to show $(3/4+\gamma)\le(1-\Omega(1))(1+\gamma)$
which holds because $\gamma\le1$.

\subsection{Proof of Lemma \ref{lem:e2-upper-bound-easy}}

Let $T$ denote the set of variables that appear in the linear polynomial
$p$ and let $U$ denote the remaining variables. Applying the same
logic as in Example \ref{exa:cor} we have
\begin{align*}
\E_{y}c_{y}(p) & =\left(\frac{1-\gamma}{2}\right)^{|T|}\left(\frac{1+\gamma}{2}\right)^{|U|}\\
 & \leq\left(\frac{1-\gamma}{2}\right)\left(\frac{1+\gamma}{2}\right)^{n-1}.
\end{align*}
The $\leq$ follows since $|T|\geq1$ and $1+\gamma>1-\gamma$ when
$\phi\in[0,\pi/4]$.

So it only remains to show $(1-\gamma)\le(1-\Omega(1))(1+\gamma)$
which holds because $\gamma\ge1/\sqrt{2}$.

\section{Boolean correlation \label{sec:B}}

In this section we prove Theorem \ref{thm:boolean-theorem-modm}.
Recall that $C_{\phi}$ is defined as the absolute value of a sum.
We need to analyze this sum more carefully, so we define it next.
\begin{defn}
$E_{\phi}(p):=\E_{x\in\zo^{n}}(-1)^{p(x)}\omega^{\sum_{i}x_{i}}$.
Note that $|E_{\phi}(p)|=C_{\phi}(p)$. 
\end{defn}

We now give an overview of the upcoming technical results. In the
proof of Theorem \ref{thm:boolean-theorem-modm}, we will use Lemma
\ref{lem:b-bounds-easy}, which relates $B_{m}(p)$ to the quantity
$|Real(E_{\phi}(p))|$ for a specific angle $\phi$, and Corollary
\ref{cor:e_phi_bounds}, which allows us to compute $|Real(E_{\phi}(s))|$
for $s=e^{2},e^{2}+e^{1}$. Together these two results will enable
us to compute $B_{m}(s)$ for $s=e^{2},e^{2}+e^{1}$. 

On the other hand, combining Lemma \ref{lem:b-bounds-easy} with Theorem
\ref{thm:main-general} lets us bound $B_{m}(p)$ when $p$ is not
symmetric, since Theorem \ref{thm:main-general} bounds $C_{\phi}(p)$
and $|Real(E_{\phi}(p))|\leq|E_{\phi}(p)|=C_{\phi}(p)$.

Proposition \ref{prop:B-expr-general} and Claims \ref{claim:mod-m-balanced},
\ref{claim:other-angles-lose} are used to prove Lemma \ref{lem:b-bounds-easy},
and Lemma \ref{lem:exact-expr-general} is needed for Corollary \ref{cor:e_phi_bounds}.

For the rest of the section, fix any odd $m\geq3$, set $\phi=2\pi/m$,
$\omega=e^{\phi\sqrt{-1}}$. We start with the following standard
fact:
\begin{prop}
\label{prop:B-expr-general}Let $b$ be the fraction of $n$-bit strings
whose weight is divisible by $m$. For any $p$,
\[
B_{m}(p)=\frac{1}{b(1-b)}\left|\frac{2}{m}\cdot\sum_{k=1}^{(m-1)/2}Real(E_{k\phi}(p))+\frac{1}{m}-b\right|
\]
where $Real(z)$ denotes the real part of the complex number $z$.
\end{prop}

\begin{proof}
Let $s(k):=\sum_{j=0}^{m}\omega^{jk}=1+\omega^{k}+\dots+\omega^{(m-1)k}$
and note that $s(k)=m$ if $k\equiv0\mod m$ and $s(k)=0$ otherwise.
Using this notation we can write
\[
B_{m}(p)=\left|\E_{x}(-1)^{p(x)}\frac{s(w(x))}{m}\cdot\frac{1}{b}-\E_{x}(-1)^{p(x)}\left(1-\frac{s(w(x))}{m}\right)\cdot\frac{1}{1-b}\right|.
\]
Collecting terms this is
\[
\left|\E_{x}(-1)^{p(x)}\left(\frac{s(w(x))}{m}\cdot\frac{1}{b}-\left(1-\frac{s(w(x))}{m}\right)\cdot\frac{1}{1-b}\right)\right|.
\]
Using the definition of $s$ this equals
\[
\left|\E_{x}(-1)^{p(x)}\left[\left(\sum_{j=1}^{m}\omega^{jw(x)}\right)\left(\frac{1}{mb}+\frac{1}{m(1-b)}\right)+\frac{1}{mb}-\left(1-\frac{1}{m}\right)\frac{1}{1-b}\right]\right|.
\]
Also, 
\[
\frac{1}{mb}-\left(1-\frac{1}{m}\right)\frac{1}{1-b}=\frac{1-mb}{mb(1-b)}.
\]
Furthermore, $\omega^{jw(x)}+\omega^{(m-j)w(x)}=2Real(\omega^{jw(x)})$
for each $j$. After factoring out $1/b(1-b)$ the result follows.
\end{proof}
Observe that in the statement of Lemma \ref{prop:B-expr-general},
if we replaced $b$ with $1/m$ then the terms that don't multiply
$\omega$ would be $0$. However, $b\ne1/m$ but it will be very close.
We use the following bound \footnote{When $m=3$ the claim says $|b-1/m|<2^{-n}$ but we do not use this.}
that's implicit in \cite{BoppanaHLV19}.
\begin{claim}
\label{claim:mod-m-balanced}$|b-1/m|<\cos(\pi/m)^{n}.$
\end{claim}

Now let $\ell_{1}\in\{\frac{m-1}{4},\frac{m+1}{4}\}$ denote the integer
closest to $\frac{m}{4}$. The next result suggests we should focus
on $Real(E_{\ell_{1}\phi}(p))$. 
\begin{claim}
\label{claim:other-angles-lose}Fix any odd $m\geq3$ and $k\in\{1...,(m-1)/2\}:k\neq\ell_{1}$.
Then for all large enough $n$ and any quadratic $p$,
\begin{align*}
\left|Real(E_{k\phi}(p))\right| & =o(\sqrt{v_{\ell_{1}\phi}}).
\end{align*}
\end{claim}

\begin{proof}
By Theorem \ref{thm:main-general}, for any $k$ it holds that
\[
C_{k\phi}(p)\leq\max_{s\in\{0,e^{1},e^{2},e^{2}+e^{1}\}}C_{k\phi}(s)\leq\max\left\{ O\left(\left(\frac{1+|\sin(k\phi)|}{2}\right)^{n/2}\right),\left(\frac{1+|\cos(k\phi)|}{2}\right)^{n/2}\right\} .
\]
Next we claim that if $k\in\{1...,(m-1)/2\},k\neq\ell_{1}$ then $\max\{|\sin(k\phi)|,|cos(k\phi)|\}<\sin(\ell_{1}\phi)$.
If this holds we can conclude since $\sqrt{v_{\ell_{1}\phi}}=\Omega((\frac{1+\sin(\ell_{1}\phi)}{2})^{n/2})$
and $|Real(E_{k\phi}(p))|\leq C_{k\phi}(p)$. 

To verify the claim, note for $k\neq\ell_{1}$, $|\sin(k\phi)|$ is
maximized when $k=\ell_{2}$, where $\ell_{2}$ denotes the second
closest integer to $m/4$. Since $m$ is odd, $\ell_{2}\in\{\frac{m-3}{4},\frac{m+3}{4}\}$
which implies $\sin(\ell_{2}\phi)<\sin(\ell_{1}\phi)$. 

And $|\cos(k\phi)|$ is maximized for $k=(m-1)/2$ and $|\cos(k\phi)|=|\cos(\pi-\pi/m)|=\cos(\pi/m)$.
We can now conclude as $\cos(\pi/m)<\sin(\ell_{1}\phi)=\sin(\pi/2\pm\pi/2m)=\cos(\pi/2m)$. 
\end{proof}
The next result, which combines Claim \ref{claim:mod-m-balanced},
\ref{claim:other-angles-lose} with Proposition \ref{prop:B-expr-general},
says we can approximate $B_{m}(p)$ using just $|Real(E_{\ell_{1}\phi}(p))|$.
\begin{lem}
\label{lem:b-bounds-easy} For all large enough $n$ and any quadratic
$p$,
\begin{align*}
\left|B_{m}(p)-\frac{2m}{m-1}\left|Real(E_{\ell_{1}\phi}(p))\right|\right|\le o(\sqrt{v_{\ell_{1}\phi}}).
\end{align*}
For $m=3$ this can be improved to
\[
\left|B_{3}(p)-3\left|Real(E_{2\pi/3}(p))\right|\right|\le O(2^{-n}).
\]
\end{lem}

\begin{proof}
By Claim \ref{claim:mod-m-balanced} and noting that $\cos(\pi/m)^{n}=o(\sqrt{v_{\ell_{1}\phi}})$
we have 
\begin{align*}
\left|\frac{1}{b(1-b)}-\frac{m^{2}}{m-1}\right|=o(\sqrt{v_{\ell_{1}\phi}}).
\end{align*}
Applying the triangle inequality and Claim \ref{claim:other-angles-lose}
we also have
\[
\bigg|\big|\sum_{k=1}^{(m-1)/2}Real(E_{k\phi}(p))\big|-\big|Real(E_{\ell_{1}\phi}(p))\big|\bigg|\leq\sum_{k\neq\ell_{1}}\big|Real(E_{k\phi}(p))\big|\leq m\cdot o(\sqrt{v_{\ell_{1}\phi}}).
\]
Inserting the previous two inequalities into Lemma \ref{prop:B-expr-general}
implies
\[
\left|B_{m}(p)-2m/(m-1)|Real(E_{\ell_{1}\phi}(p))|\right|\leq O(m)o(\sqrt{v_{\ell_{1}\phi}}).
\]
We can now conclude since we consider $m$ fixed.
\end{proof}
We are naturally interested in computing $B_{m}(s)$ for $s=e^{2},e^{2}+e^{1}$
and the next lemma allows us to do so by giving an expression for
$E_{\ell_{1}\phi}(s)$. In Section \ref{sec:symmetry} we determined
$C_{\ell_{1}\phi}(s)=|E_{\ell_{1}\phi}(s)|$, but this no longer suffices
as we need to understand the angle of $E_{\ell_{1}\phi}(s)$ in order
to compute $|Real(E_{\ell_{1}\phi}(s))|$.
\begin{lem}
\label{lem:exact-expr-general} For any $k\in\{1,2,\dots,m-1\}$ we
have:
\begin{align*}
E_{k\phi}(e^{2}) & =2^{-(n+1)}\left[(1+i)(1-i\omega^{k})^{n}+(1-i)(1+i\omega^{k})^{n}\right],\\
E_{k\phi}(e^{2}+e^{1}) & =2^{-(n+1)}\left[(1-i)(1-i\omega^{k})^{n}+(1+i)(1+i\omega^{k})^{n}\right].
\end{align*}
\end{lem}

\begin{proof}
We prove Item 1. Since $(-1)^{e^{2}(x)}=(-1)^{\binom{w(x)}{2}}$ we
can write
\begin{align*}
E_{k\phi}(e^{2}) & =\sum_{j=0}^{n}\binom{n}{j}(-1)^{\binom{j}{2}}\omega^{kj}.
\end{align*}
We also have
\begin{align*}
\sum_{j=0\bmod4}\binom{n}{j}\omega^{kj} & =\sum_{j=0}^{n}\binom{n}{j}\omega^{kj}\left(\frac{1+i^{j}}{2}\right)\left(\frac{1+(-1)^{j}}{2}\right)\\
\sum_{j=2\bmod4}\binom{n}{j}\omega^{kj} & =\sum_{j=0}^{n}\binom{n}{j}\omega^{kj}\left(\frac{1-i^{j}}{2}\right)\left(\frac{1+(-1)^{j}}{2}\right).
\end{align*}
So this implies
\[
\sum_{j=0\bmod4}\binom{n}{j}\omega^{kj}-\sum_{j=2\bmod4}\binom{n}{j}\omega^{kj}=\frac{1}{2}\left[(1+\omega^{k}i)^{n}+(1+\omega^{k}(-i))^{n}\right].
\]
Doing the analogous for $j=1,3\mod4$ gives
\[
\sum_{j=1\bmod4}\binom{n}{j}\omega^{kj}-\sum_{j=3\bmod4}\binom{n}{j}\omega^{kj}=\frac{1}{2}\left[-i(1+\omega^{k}i)^{n}+i(1+\omega^{k}(-i))^{n}\right].
\]
The proof of Item 2 is similar.
\end{proof}
The next result reduces the problem of computing $|Real(E_{\ell_{1}\phi}(s)|$
to the problem of computing $|\cos(\chi\pm\pi/4)|$ for a certain
angle $\chi$. The angle $\chi\pm\pi/4$ arises because it is the
angle of the vector $(1\pm i)(1-i\omega^{\ell_{1}})^{n}$, which is
the dominant term in the previous expressions for $E_{\ell_{1}\phi}(s)$.
The last equality below then allows us to relate $|Real(E_{\ell_{1}\phi}(s)|$
to $\sqrt{v_{\ell_{1}\phi}}$.
\begin{cor}
\label{cor:e_phi_bounds} Let $\chi=\frac{n\pi}{4m},-\frac{n\pi}{4m}$
when $\ell_{1}=\frac{m+1}{4},\frac{m-1}{4}$ respectively. Let $\gamma=\sqrt{2}|1-i\omega^{\ell_{1}}|^{n}$.
For all large enough $n$, the following holds:
\begin{enumerate}
\item $\left|2^{n+1}|Real(E_{\ell_{1}\phi}(e^{2}))|-|\cos(\chi+\pi/4)|\gamma\right|=o(1)$
\item $\left|2^{n+1}|Real(E_{\ell_{1}\phi}(e^{2}+e^{1}))|-|\cos(\chi-\pi/4)|\gamma\right|=o(1),$
\item $\left|2^{n+1}\sqrt{v_{\ell_{1}\phi}}-\gamma\right|=o(1).$
\end{enumerate}
\end{cor}

\begin{proof}
We show the first equality when $\ell_{1}=\frac{m+1}{4}$. The $\ell_{1}=\frac{m-1}{4}$
case is symmetrical.

By definition $\omega^{\ell_{1}}=e^{\sqrt{-1}(2\pi/m)(m+1)/4}=e^{\sqrt{-1}(\pi/2+\pi/2m)}$,
hence $-i\omega^{\ell_{1}}=e^{\sqrt{-1}(\pi/2m)}$. This implies $(1-i\omega^{\ell_{1}})=|1-i\omega^{\ell_{1}}|e^{\sqrt{-1}(\pi/4m)}$.
Additionally, $1+i=\sqrt{2}e^{\sqrt{-1}(\pi/4)}$. So then
\begin{align*}
(1+i)(1-i\omega^{\ell_{1}})^{n} & =\sqrt{2}e^{\sqrt{-1}(\pi/4)}\cdot|1-i\omega^{\ell_{1}}|^{n}e^{\sqrt{-1}(\pi/4m)n}\\
 & =\gamma e^{\sqrt{-1}(n\pi/4m+\pi/4)}.
\end{align*}
We can now conclude by Lemma \ref{lem:exact-expr-general}, the fact
$|Real(e^{\sqrt{-1}\phi})|=|\cos\phi|$ for any $\phi$, and noting
$|1+i\omega^{\ell_{1}}|^{n}=o(1)$ since $|1+i\omega^{\ell_{1}}|<1$
when $m$ is odd. The second inequality is done similarly. 

The third inequality follows by Lemma \ref{lem:exact-expr-general},
the facts $|E_{\ell_{1}\phi}(p)|=C_{\ell_{1}\phi}(p)$, $|1+i\omega^{\ell_{1}}|^{n}=o(1)$,
and since when $s=e^{2},e^{2}+e^{1}$, $\left|C_{\ell_{1}\phi}(s)-\sqrt{v_{\ell_{1}\phi}}\right|\leq o(1)$
by Lemma \ref{lem:=00005BContributions-of-symmetric}.
\end{proof}

\subsection{Proof of Theorem \ref{thm:boolean-theorem-modm}}

\subsubsection{Proof of Item 1}

First we prove the upper bound. Lemma \ref{lem:b-bounds-easy} implies
that 
\[
B_{m}(p)\leq2m/(m-1)\left|Real(E_{\ell_{1}\phi}(p))\right|+o(\sqrt{v_{\ell_{1}\phi}}).
\]
The upper bound now follows since $\left|Real(E_{\ell_{1}\phi}(p))\right|\leq\left|E_{\ell_{1}\phi}(p)\right|=C_{\ell_{1}\phi}(p)\leq(1+o(1))\sqrt{v_{\ell_{1}\phi}}$.
The last inequality holds by Theorem \ref{thm:main-general}.

Next we prove the lower bound by showing 
\begin{equation}
\max_{s\in\{e^{2},e^{2}+e^{1}\}}B_{m}(s)\geq(2m/(m-1)-o(1))\sqrt{\frac{v_{\ell_{1}\phi}}{2}}.\label{eq:boolean_lower_bound}
\end{equation}
Lemma \ref{lem:b-bounds-easy} implies that
\begin{align*}
B_{m}(s) & \geq2m/(m-1)\left|Real(E_{\ell_{1}\phi}(s))\right|-o(\sqrt{v_{\ell_{1}\phi}}).
\end{align*}
Then we claim that for either $s=e^{2}$ or $s=e^{2}+e^{1}$,
\begin{align*}
\left|Real(E_{\ell_{1}\phi}(s))\right| & \geq(1-o(1))\sqrt{\frac{v_{\ell_{1}\phi}}{2}}.
\end{align*}
The previous two inequalities imply Equation \ref{eq:boolean_lower_bound}.

To verify the claim, note that since $\cos(\pi/4)=1/\sqrt{2}$, at
least one of the next two inequalities hold for any angle $\chi$:
\begin{align*}
\cos(\chi+\pi/4) & \geq1/\sqrt{2},\\
\cos(\chi-\pi/4) & \geq1/\sqrt{2}.
\end{align*}
We then conclude by Corollary \ref{cor:e_phi_bounds}.

\subsubsection{Proof of Item 2}

We present the $n\equiv3m\bmod4m,\ell_{1}=\frac{m+1}{4}$ case. In
the proof we show that $E_{\ell_{1}}(e^{2})$ is essentially real,
which means $|Real(E_{\ell_{1}}(e^{2}))|$ equals $\sqrt{v_{\ell_{1}\phi}}$
by Corollary \ref{cor:e_phi_bounds}. On the other hand, for any non-symmetric
$p$, $C_{\ell_{1}\phi}(p)$ is a constant factor smaller than $\sqrt{v_{\ell_{1}\phi}}$
by Theorem \ref{thm:main-general}. This suffices as $|E_{\ell_{1}\phi}(p)|=C_{\ell_{1}\phi}(p)$,
and note the angle of $E_{\ell_{1}\phi}(p)$ does not even matter.

So first we show
\[
B_{m}(e^{2})\geq(2m/(m-1)-o(1))\sqrt{v_{\ell_{1}\phi}}.
\]
This follows by Lemma \ref{lem:b-bounds-easy} and the claim that
\begin{align*}
\left|Real(E_{\ell_{1}\phi}(e^{2}))\right| & \geq(1-o(1))\sqrt{v_{\ell_{1}\phi}}.
\end{align*}

To verify the claim, note when $n\equiv3m\bmod4m$, $n\pi/4m=(3m+k4m)\pi/4m\equiv3\pi/4+k\pi\bmod2\pi$
for some integer $k$. Hence $\cos(n\pi/4m+\pi/4)=\cos((k+1)\pi)=\pm1$.
We then conclude by Corollary \ref{cor:e_phi_bounds}. Note $\cos(n\pi/4m-\pi/4)=0$,
so $B_{m}(e^{2}+e^{1})<B_{m}(e^{2})$.

On the other hand, for any $p\neq e^{2},e^{2}+e^{1}$ we show 
\[
B_{m}(p)\leq2m/(m-1)\sqrt{1-\Omega(\sin(\ell_{1}\phi)-\cos(\ell_{1}\phi))}\cdot\sqrt{v_{\ell_{1}\phi}}.
\]
 This follows by Lemma \ref{lem:b-bounds-easy} and Theorem \ref{thm:main-general}
which states
\[
C_{\ell_{1}\phi}(p)\leq\sqrt{1-\Omega(\sin(\ell_{1}\phi)-\cos(\ell_{1}\phi))}\cdot\sqrt{v_{\ell_{1}\phi}}.
\]
This yields the desired inequality since $\left|Real(E_{\ell_{1}\phi}(p))\right|\leq C_{\ell_{1}\phi}(p)$.

If $p=e^{1},0$ we show 

\begin{equation}
\max_{s\in\{0,e^{1}\}}B_{m}(s)\leq(2m/(m-1))\cdot o(\sqrt{v_{\ell_{1}}}).\label{eq:degree_1_bounds}
\end{equation}
This follows by Lemma \ref{lem:b-bounds-easy} and noting for $s=e^{1},0$,
$C_{\ell_{1}\phi}(s)=(\frac{1+\cos(\ell_{1}\phi)}{2})^{n/2}=o(\sqrt{v_{\ell_{1}\phi}})$
since $\cos(\ell_{1}\phi)<\sin(\ell_{1}\phi).$

The $n\equiv3m,\ell_{1}=\frac{m-1}{4}$ case is similar except we
use $e^{2}+e^{1}$ instead of $e^{2}$. The $n\equiv m$ cases are
analogous. 

\subsubsection{Proof of Item 3}

We present the $n\equiv0\bmod4m,\ell_{1}=\frac{m+1}{4}$ case. First
note that Equations \ref{eq:boolean_lower_bound} and \ref{eq:degree_1_bounds}
imply it suffices to prove $\max_{s\in\{e^{2},e^{2}+e^{1}\}}B_{m}(s)<B_{m}(q)$
for some non-symmetric $q$. We will show that $E_{\ell_{1}\phi}(e^{2}),E_{\ell_{1}\phi}(e^{2}+e^{1})$
are both maximally imaginary as allowed by Equation \ref{eq:boolean_lower_bound}.
Next, consider $q:=x_{1}+e^{2}(x_{2,}\dots,x_{n})$. $C_{\ell_{1}\phi}(q)$
is close to, but less than $C_{\ell_{1}\phi}(s)$ for $s=e^{2},e^{2}+e^{1}$.
However, $E_{\ell_{1}\phi}(q)$ will be more real which is enough
to compensate for this difference and show that $\left|Real(E_{\ell_{1}\phi}(s))\right|<|Real(E_{\ell_{1}\phi}(q))|$.

So first we show that for either $s=e^{2},e^{2}+e^{1}$,

\[
B_{m}(s)\leq(2m/(m-1)+o(1))\cdot\sqrt{\frac{v_{\ell_{1}\phi}}{2}}.
\]
This follows by Lemma \ref{lem:b-bounds-easy} and the claim that
for either $s=e^{2},e^{2}+e^{1},$
\begin{align*}
\left|Real(E_{\ell_{1}\phi}(s))\right| & \leq(1+o(1))\sqrt{\frac{v_{\ell_{1}\phi}}{2}}.
\end{align*}
To verify the claim, since $n\equiv0\bmod4m$, then $n\pi/4m\equiv k\pi\bmod2\pi.$
Hence $\cos(n\pi/4m\pm\pi/4m)=\pm1/\sqrt{2}$. We then conclude by
Corollary \ref{cor:e_phi_bounds}.

On the other hand, we show that 
\[
B_{m}(q)>(2m/(m-1)-o(1))\cdot\frac{(1+\tan(\pi/4m))\sqrt{v_{\ell_{1}\phi}}}{\sqrt{2}}.
\]
 Note $1+\tan(\pi/4m)>1$ for $m\geq3$. The inequality holds by Lemma
\ref{lem:b-bounds-easy} and the claim

\[
\left|Real(E_{\ell_{1}\phi}(q))\right|\geq(1-o(1))\cdot\frac{(1+\tan(\pi/4m))\sqrt{v_{\ell_{1}\phi}}}{\sqrt{2}}.
\]

To show the claim, we start by rewriting $E_{\ell_{1}\phi}(q)$ by
conditioning on $x_{1}$ (below $e^{2}$ is on $n-1$ variables):
\begin{align*}
E_{\ell_{1}\phi}(q) & =\frac{(1-\omega^{\ell_{1}})}{2}E_{\ell_{1}\phi}(e^{2}).
\end{align*}
An analogous version of Corollary \ref{cor:e_phi_bounds} Item 1 holds
for $e^{2}$ on $n-1$ variables:
\begin{align*}
\left|2^{n}|Real(E_{\ell_{1}\phi}(e^{2}))|-\left|\cos\left(\frac{(n-1)\pi}{4m}+\frac{\pi}{4}\right)\right|\frac{\gamma}{|1-i\omega^{\ell_{1}}|}\right| & =o(1).
\end{align*}
Since $-\omega^{\ell_{1}}=e^{\sqrt{-1}(-\pi/2+\pi/2m)}$ we have $(1-\omega^{\ell_{1}})=|1-\omega^{\ell_{1}}|e^{\sqrt{-1}(-\pi/4+\pi/4m)}.$
Combining this with the previous equality implies that 
\begin{align*}
\left|2^{n+1}|Real(E_{\ell_{1}\phi}(q))|-\left|\cos\left(\frac{(n-1)\pi}{4m}+\frac{\pi}{4m}\right)\right|\frac{|1-\omega^{\ell_{1}}|}{|1-i\omega^{\ell_{1}}|}\gamma\right| & =o(1)\\
\iff\left|2^{n+1}|Real(E_{\ell_{1}\phi}(q))|-\frac{|1-\omega^{\ell_{1}}|}{|1-i\omega^{\ell_{1}}|}\gamma\right| & =o(1).
\end{align*}
The $\iff$ follows as $\cos(n\pi/4m)=\pm1$ when $n\equiv0\bmod4m$. 

To conclude, by Corollary \ref{cor:e_phi_bounds} it suffices to show
\[
\frac{1+\tan(\pi/4m)}{\sqrt{2}}=\frac{|1-\omega^{\ell_{1}}|}{|1-i\omega^{\ell_{1}}|}.
\]

Using the identity $|1+e^{\sqrt{-1}\phi}|=2|cos(\phi/2)|$, we have
$|1-i\omega^{\ell_{1}}|=2\cos(\pi/4m)$ and $|1-\omega^{\ell_{1}}|=2|\cos(-\pi/4+\pi/4m)|=2\cos(\pi/4-\pi/4m)=\sqrt{2}(\cos(\pi/4m)+\sin(\pi/4m))$
where the last step holds as $cos(a-b)=\cos a\cos b+\sin a\sin b$.
Hence the equality holds.

The $n\equiv0$, $\ell_{1}=\frac{m-1}{4}$ case is similar except
$q$ will be $e^{2}(x_{2},\dots,x_{n})$ instead. The $n\equiv2m$
cases are analogous. 

\section{\label{sec:symm_vs_mod}Symmetric correlates poorly with mod $m$}

For completeness, we show that symmetric polynomials mod 2 correlate
poorly with the complex mod $m$ function. To get a sense of the parameters
below, fix $m=3$ and apply the identities $\cos x\leq1-x^{2}/6$
and $(1-x)^{n}\leq e^{-xn}$. This yields $C_{\phi}(s)\leq O(d)2^{-\Omega(n/d^{2})}$,
so if Conjecture \ref{conj:sym-best} were true this would imply exponentially
small correlation bounds for any $O(\log n)$ degree polynomial -
a long-standing open problem.
\begin{thm}
\label{thm:symm_corr_mod3}Let $\phi=2\pi k/m$ for some odd $m$
and $k\in\{1,\dots m-1\}$. Then for any degree $d$ symmetric polynomial
$s$,
\[
C_{\phi}(s)\leq2md\cdot\cos\left(\frac{\pi}{2md}\right)^{n}.
\]
\end{thm}

\begin{proof}
Let $\delta$ be an integer such that $2^{\delta-1}\leq d<2^{\delta}$.
It is shown in \cite{BGL06} that $s(x)$ is determined by the weight
of $x$ mod $2^{\delta}$. Hence we can write
\[
(-1)^{s(x)}=\sum_{i=0}^{2^{\delta}-1}c_{i}\mathbf{1}_{w(x)\equiv i\bmod2^{\delta}}
\]
where $c_{i}\in\{-1,1\}$ for each $i$. Then we can write the correlation
as
\begin{align*}
C_{\phi}(s) & =\left|\E_{x}[Mod_{\phi}(x)\cdot\sum_{i=0}^{2^{\delta}-1}c_{i}\mathbf{1}_{w(x)\equiv i\bmod2^{\delta}}]\right|\\
 & =\left|\sum_{i=0}^{2^{\delta}-1}\E_{x}\left[Mod_{\phi}(x)\cdot c_{i}\mathbf{1}_{w(x)\equiv i\bmod2^{\delta}}\right]\right|.
\end{align*}
Letting $\omega=e^{\sqrt{-1}\cdot2\pi/m}$, for any $i$ we have 
\begin{align*}
\E_{x}\left[Mod_{\phi}(x)\cdot\mathbf{1}_{w(x)\equiv i\bmod2^{\delta}}\right] & =\sum_{j=0}^{m-1}\omega^{(i+j2^{\d})k}\P_{x}[w(x)\equiv i+j2^{\delta}\bmod m2^{\delta}].
\end{align*}
We next use a slightly generalized version of Claim \ref{claim:mod-m-balanced}:
\end{proof}
\begin{claim}
For any $k,m$, $\left|\P_{x}[w(x)\equiv k\bmod m]-1/m\right|\leq\cos(\pi/m){}^{n}.$
\end{claim}

\begin{proof}
Combining this with the fact $\sum_{j=0}^{m-1}\omega^{(i+j2^{\d})k}=0$
implies that 
\[
\left|\E_{x}\left[Mod_{\phi}(x)\cdot\mathbf{1}_{w(x)\equiv i\bmod2^{\delta}}\right]\right|\leq m(\cos(\pi/m2^{\delta}))^{n}.
\]
Hence
\[
C_{\phi}(s)\leq\sum_{i=0}^{2^{\delta}-1}\left|\E_{x}\left[Mod_{\phi}(x)\cdot c_{i}\mathbf{1}_{w(x)\equiv i\bmod2^{\delta}}\right]\right|\leq m2^{\delta}\cos(\pi/m2^{\delta}))^{n}.
\]
We can now conclude the proof since $2^{\d}\leq2d$.
\end{proof}

\section{\label{sec:structured-cubic} Structured cubic loses to quadratic}

In this section we show that any cubic polynomial with a symmetric
degree 3 part has correlation that is a constant factor worse than
the optimal achieved by quadratic polynomials. 
\begin{thm}
\label{thm:symm_deg3_part_loses}Suppose $t=e^{3}+q$ for some arbitrary
quadratic $q$. Then for any $\phi$,
\[
C_{\phi}(t)\leq(1-\Omega(1))\max_{s\in\{0,e^{1},e^{2},e^{2}+e^{1}\}}C_{\phi}(s).
\]
\end{thm}

We first show that cubic symmetric polynomial $e^{3}$ has worse correlation
than the optimal quadratic symmetric. We prove this by applying the
derivative framework from Section \ref{sec:Derivatives}. We analyze
for every direction $y$ what the derivative $e_{y}^{3}$ will be
and use this to bound the contribution $|c_{y}(e^{3})|$ in Lemma
\ref{lem:e3_contribution_bounds}. 

Next we show that $t=e^{3}+q$ can only have worse correlation than
$e^{3}$ for any quadratic $q$. We do this in Lemma \ref{lem:cubic_symmmetric_wins}
by showing that for any direction $y$, adding the derivative $q_{y}$
(which will be linear) to $e_{y}^{3}$ can only decrease the contribution.
In other words, we show $|c_{y}(t)|\leq|c_{y}(e^{3})|$ for every
$y$. 
\begin{lem}
\label{lem:e3_contribution_bounds}For any $y$,
\begin{enumerate}
\item If $w(y)\in E$ then
\[
|c_{y}(e^{3})|\leq\frac{|\sigma|^{w(y)}+|\gamma|^{w(y)}}{2}.
\]
\item If $w(y)\in O$ then 
\[
|c_{y}(e^{3})|\leq\frac{2^{w(y)}}{2^{n-1}}.
\]
\end{enumerate}
\end{lem}

\begin{lem}
\label{lem:cubic_symmmetric_wins} Suppose $t=e^{3}+q$ for some arbitrary
quadratic $q$. Then for any $y$, 
\[
|c_{y}(t)|\leq|c_{y}(e^{3})|.
\]
\end{lem}

The previous two lemmas imply Theorem \ref{thm:symm_deg3_part_loses}.
\begin{proof}[Proof of Theorem \ref{thm:symm_deg3_part_loses} assuming Lemmas \ref{lem:e3_contribution_bounds},
\ref{lem:cubic_symmmetric_wins}. ]

By Lemmas \ref{lem:e3_contribution_bounds}, \ref{lem:cubic_symmmetric_wins}
we have
\begin{align*}
C_{\phi}^{2}(t) & \leq\sum_{y:w(y)\in E}\frac{|\sigma|^{w(y)}+|\gamma|^{w(y)}}{2}+\sum_{y:w(y)\in O}2^{-(n-w(y)-1)}\\
 & =\frac{(1+|\sigma|)^{n}+(1-|\sigma|)^{n}}{4}+\frac{(1+|\gamma|)^{n}+(1-|\gamma|)^{n}}{4}+\frac{3^{n}-1}{2^{n}}.
\end{align*}
The $=$ follows by Claim \ref{claim:=00005BOdd-even-sum-claim=00005D}.
Next note that for any $\phi$, $\max\{1+|\sigma|,1+|\gamma|\}\geq1+1/\sqrt{2}>3/2$.
Suppose $\phi$ is such that $|\sigma|>|\gamma|$. Then 
\[
C_{\phi}^{2}(t)\leq2^{-n}\frac{(1+o(1))(1+|\sigma|)^{n}}{4}.
\]
On the other hand by Theorem \ref{thm:main-general} we know that
\[
\max_{s\in\{e^{2},e^{2}+e^{1}\}}C_{\phi}^{2}(s)\geq2^{-n}\frac{(1+|\sigma|)^{n}}{2}.
\]
Now suppose $\phi$ is such that $|\sigma|\leq|\gamma|$. Then 
\[
C_{\phi}^{2}(t)\leq2^{-n}\frac{(2+o(1))(1+|\gamma|)^{n}}{4}.
\]
However by Theorem \ref{thm:main-general},
\[
\max_{s\in\{0,e^{1}\}}C_{\phi}^{2}(s)=2^{-n}(1+|\gamma|)^{n}.
\]
\end{proof}

\subsection{Proof of Lemma \ref{lem:e3_contribution_bounds}}

We first list some preliminary results we will need. The following
is a standard fact we state without proof.
\begin{fact}
\label{fact:quadratic_symm_bias} Let $s$ denote either $e^{2},e^{2}+e^{1}$
on $n$ variables, and let $\ell$ denote an arbitrary linear polynomial.
Then $|bias((-1)^{s+\ell})|\leq2^{-(n-1)/2}$.
\end{fact}

Below and for the remainder of the section, we let $V_{1},V_{0}\subseteq[n]$
denote the indices of the 1, 0-variables respectively with respect
to a fixed direction $y$. 

The next result says that if the bias of $p_{y}$ is small after an
arbitrary restriction to the 1-variables, then $|c_{y}(p)|$ must
be small.
\begin{prop}
\label{prop:small_bias_small_corr}Fix some polynomial $p$ and direction
$y\in\zo^{n}$. Suppose for any restriction $r\in\zo^{|V_{1}|}$ of
the 1-variables, 
\[
\left|\E_{x:x^{V_{1}}=r}(-1)^{p_{y}(x)}\right|\leq\delta.
\]
Then 
\[
|c_{y}(p)|\leq\delta.
\]
\end{prop}

\begin{proof}
We have
\begin{align*}
c_{y}(p) & =\E_{x}[(-1)^{p_{y}(x)}Mod_{\phi,y}(x)]\\
 & =\E_{x^{V_{1}}}[Mod_{\phi,y}(x)\cdot\E_{x^{V_{0}}}[(-1)^{p_{y}(x)}]]\\
 & \leq\d.
\end{align*}
 The second $=$ follows since $Mod_{\phi,y}(x)$ only depends on
the 1-variables. The $\leq$ follows since $|Mod_{\phi,y}(x)|=1$
and by the hypothesis on $p_{y}$.
\end{proof}
Next we characterize the derivatives of $e^{3}$ which depend on the
weight of $y$ $\bmod4$. We abuse notation and let $e^{i}(V_{j})$
denote the polynomial $e^{i}$ defined on the variables indexed by
$V_{j}$. 
\begin{prop}
\label{prop:derivative_e3} Fix any direction $y\in\zo^{n}$ and consider
the derivative $e_{y}^{3}$.
\begin{enumerate}
\item If $w(y)\equiv0\bmod4$ then
\[
e_{y}^{3}=e^{1}(V_{1})+e^{1}(V_{1})e^{1}(V_{0}).
\]
\item If $w(y)\equiv2\bmod4$ then
\[
e_{y}^{3}=e^{1}(V_{1})e^{1}(V_{0})+e^{1}(V_{0}).
\]
\item If $w(y)\equiv1\bmod4$ then
\[
e_{y}^{3}=e^{2}(V_{1})+e^{2}(V_{0}).
\]
\item If $w(y)\equiv3\bmod4$ then
\[
e_{y}^{3}=(e^{2}+e^{1})(V_{1})+(e^{2}+e^{1})(V_{0})+1.
\]
\end{enumerate}
\end{prop}

\begin{proof}
We can write $e^{3}=e^{3}(V_{1})+e^{2}(V_{1})e^{1}(V_{0})+e^{1}(V_{1})e^{2}(V_{0})+e^{3}(V_{0})$.
Firstly note the term $e^{3}(V_{0})$ does not affect $e_{y}^{3}$.
Secondly, the term $e^{1}(V_{1})e^{2}(V_{0})$ only contributes $e^{2}(V_{0})$
to $e_{y}^{3}$ when $|V_{1}|=w(y)$ is odd.

Thirdly, we deal with $e^{2}(V_{1})e^{1}(V_{0})$. Note that $e^{1}(V_{0})$
has a coefficient of $\binom{w(y)}{2}$ in $e_{y}^{3}$, which is
odd when $w(y)\equiv2,3\bmod4$. Now let $x_{i}$ denote a 1-variable.
Then $x_{i}e^{1}(V_{0})$ has a coefficient of $\binom{w(y)-1}{1}$,
hence $e^{1}(V_{1})e^{1}(V_{0})$ appears when $w(y)$ is even. 

Lastly, we deal with $e^{3}(V_{1})$. Note $x_{i}$ has a coefficient
of $\binom{w(y)-1}{2}$, hence $e^{1}(V_{1})$ appears if $w(y)\equiv0,3\bmod4$.
Let $x_{j}$ denote a second 1-variable. Then $x_{i}x_{j}$ has a
coefficient of $\binom{w(y)-2}{1}$ hence $e^{2}(V_{1})$ appears
if $w(y)$ is odd. The constant 1 has a coefficient of $\binom{w(y)}{3}$
which is odd when $w(y)\equiv3\bmod4$.
\end{proof}

\paragraph{Proof of Lemma \ref{lem:e3_contribution_bounds} }

Suppose $w(y)\equiv0\bmod4$. By Proposition \ref{prop:derivative_e3},
if $x^{V_{0}}\in E$ then $e_{y}^{3}=e^{1}(V_{1})$. If $x^{V_{0}}\in O$
then $e_{y}^{3}=0$. Hence
\begin{align*}
c_{y}(e^{3}) & =2^{-n}(\sum_{x:x^{V_{0}}\in E}\sigma^{w(y)}+\sum_{x:x^{V_{0}}\in O}\gamma^{w(y)})\\
 & =\frac{\sigma^{w(y)}+\gamma^{w(y)}}{2}.
\end{align*}
The $w(y)\equiv2\bmod4$ case is similar. If $x^{V_{0}}\in E$ then
$e_{y}^{3}=0$. Otherwise, $e_{y}^{3}=e^{1}(V_{0})+1$. Hence $c_{y}(e^{3})=\frac{-\sigma^{w(y)}+\gamma^{w(y)}}{2}$.
This concludes the $w(y)\in E$ case. 

Now suppose $w(y)\in O$. Fact \ref{fact:quadratic_symm_bias} implies
that for $s=e^{2}(V_{0}),(e^{2}+e^{1})(V_{0})$, $|bias((-1)^{s})|\leq2^{-(n-w(y)-1)}$.
Since $e_{y}^{3}$ is disjoint on $V_{0},V_{1}$, Proposition \ref{prop:small_bias_small_corr}
implies that $|c_{y}(e^{3})|\leq2^{-(n-w(y)-1)}.$ 

\subsection{Proof of Lemma \ref{lem:cubic_symmmetric_wins}}

Suppose that $t=e^{3}+q$ for some quadratic $q$. Note that for any
direction $y$, $t_{y}$ has the same quadratic terms as $e_{y}^{3}$
and $q_{y}$ only affects the linear terms in $p_{y}$. Let us write
$q_{y}=u(V_{1})+v(V_{0})$, where $u(V_{1}),v(V_{0})$ are linear
polynomials over the $1,0$-variables respectively.

First suppose $y\equiv0\bmod4$. We now consider restricting the 1-variables.
If $x^{V_{1}}\in E$ then $t_{y}^{3}=c+v(V_{0})$ where $c$ is some
constant. If $x^{V_{1}}\in O$ then $t_{y}^{3}=c+(e^{1}+v)(V_{0})$.
Note that if $0\neq v(V_{0})\neq e^{1}(V_{0})$, then the bias of
the restricted function will be 0 for both cases. Hence by Proposition
\ref{prop:small_bias_small_corr}, $c_{y}(t)=0$ and we are done.
If $v(V_{0})=e^{1}(V_{0})$ then this is symmetrical to when $v(V_{0})=0$.
Hence we can assume that $v(V_{0})=0$.

From here, we switch back to restricting the 0-variables. If $x^{V_{0}}\in E$
then $e_{y}^{3}=(e^{1}+u)(V_{1})$, and if $x^{V_{0}}\in O$ then
$e_{y}^{3}=u(V_{1})$. Suppose $u(V_{1})$ contains $k$ variables.
Then $|c_{y}(t)|\leq|\sigma|^{w(y)-k}|\gamma|^{k}$ whenever $x^{V_{0}}\in E$
and $|c_{y}(t)|\leq|\sigma|^{k}|\gamma|^{w(y)-k}$ otherwise. Hence
\begin{align*}
|c_{y}(t)| & \leq\frac{|\sigma|^{w(y)-k}|\gamma^{k}|+|\sigma|^{k}|\gamma|^{w(y)-k}}{2}.
\end{align*}
Assume that $|\sigma|>|\gamma|$ (the other case is similar). We can
now conclude as
\begin{align*}
\frac{|\sigma|^{w(y)-k}|\gamma|^{k}+|\sigma|^{k}|\gamma|^{w(y)-k}}{2} & \leq\frac{|\sigma|^{w(y)}+|\gamma|^{w(y)}}{2}\\
\iff\frac{|\gamma|^{w(y)-k}(|\sigma|^{k}-|\gamma|^{k})}{2} & \leq\frac{|\sigma|^{w(y)-k}(|\sigma|^{k}-|\gamma|^{k})}{2}\\
\iff|\gamma| & \leq|\sigma|.
\end{align*}
The $w(y)\equiv2\bmod4$ case is analogous. 

Now suppose $w(y)\equiv1\bmod4$. After an arbitrary restriction to
$x^{V_{1}}$, we have $e_{y}^{3}=e^{2}(V_{0})+v(V_{0})+c$ for some
constant $c$. Fact \ref{fact:quadratic_symm_bias} implies that $|bias((-1)^{e_{y}^{3}})|\leq2^{-(n-w(y)-1)}$
after any restriction to $x^{V_{1}}.$ We can now conclude by applying
Proposition \ref{prop:small_bias_small_corr}. The $w(y)\equiv3\bmod4$
case is analogous.

\paragraph{Acknowledgment.}

We are grateful to Brenden Collins for collaborating during the initial
stages of this project.

\bibliographystyle{alpha}
%\bibliography{omniBib}
\newcommand{\etalchar}[1]{$^{#1}$}
\def\cprime{$'$}

\end{document}